\documentclass[11pt, leqno]{article}
\usepackage[paperwidth=8.5in,paperheight=11in,top=1in, bottom=1in, left=0.75in, right=0.75in]{geometry} 

\usepackage{amsfonts,amsthm,amssymb,amsmath}
\allowdisplaybreaks
\usepackage{xcolor}
\usepackage{tikz}
\usepackage{bm}
\usepackage{epstopdf}
\usepackage{graphicx}

\usepackage{array}
\usepackage{mathtools}
\mathtoolsset{showonlyrefs=true}
\numberwithin{equation}{section}

\usepackage{multirow}
\usepackage[authoryear]{natbib}

\usepackage{hyperref} % enables and customizes hyperlinks
\hypersetup{
	colorlinks   = true,
	citecolor    = blue,
	linkcolor    = red,
	urlcolor     = magenta
}

\newtheorem{theorem}{Theorem}[section]

\newtheorem{proposition}[theorem]{Proposition}
\newtheorem{remark}[theorem]{Remark}

\newtheorem{problem}[theorem]{Problem}
\newtheorem{definition}[theorem]{Definition}
\newtheorem{corollary}[theorem]{Corollary}

\newcommand\sig{\sigma}

\newcommand\Lam{\Lambda}
\newcommand\gam{\gamma}

\newcommand\lam{\lambda}
\newcommand\del{\delta}

\newcommand\Ac{\mathcal{A}}

\newcommand\Fc{\mathcal{F}}

\newcommand{\Eb}{\mathbb{E}}
\newcommand{\Fb}{\mathbb{F}}

\newcommand{\Pb}{\mathbb{P}}

\newcommand{\Rb}{\mathbb{R}}

\newcommand{\dd}{\mathrm{d}}

\newcommand{\wh}{\widehat}
\newcommand{\wt}{\widetilde}

\newcommand{\wone}{W^{(1)}}
\newcommand{\wtwo}{W^{(2)}}
\newcommand{\xone}{X^{(1)}}
\newcommand{\dione}{D^{(1)}}

\begin{document}

\title{ A Perturbation Approach to \\Optimal Investment, Liability Ratio, and Dividend Strategies}

\author{Zhuo Jin\thanks{Centre for Actuarial Studies, Department of Economics, University of Melbourne, Australia. Email: zjin@unimelb.edu.au. This author is partially supported by the Research Grants Council of the Hong Kong Special Administrative Region (No. 17330816).}
\and
Zuo Quan Xu\thanks{Department of Applied Mathematics, Hong Kong Polytechnic University, Hong Kong, China. Email: maxu@polyu.edu.hk. 
This author is partially supported by by the National Natural Science Foundation of China  (No. 11971409) and Hong Kong GRF (No. 15204216 and No. 15202817). 
}
\and
Bin Zou\thanks{Corresponding author.   Department of Mathematics, University of Connecticut, 341 Mansfield Road U1009, Storrs, Connecticut 06269-1009, USA. Email: bin.zou@uconn.edu. Phone: +1-860-486-3921. 
This author is partially supported by a start-up grant from the University of Connecticut.}
}

\date{This Version: May 26, 20201\\Forthcoming in \emph{Scandinavian Actuarial Journal} }
\maketitle

\vspace{-1cm}
\begin{abstract}
\normalsize
We study an optimal dividend problem for an insurer who simultaneously controls 
investment weights in a financial market, liability ratio in the insurance business, and dividend payout rate. 
The insurer seeks an optimal strategy to maximize her expected utility of dividend payments over an infinite horizon. 
By applying a perturbation approach, we obtain the optimal strategy and the value function in closed form for log and power utility.  
We conduct an economic analysis to investigate the impact of various model parameters and risk aversion on the insurer's optimal strategy. 
\end{abstract}

\noindent
\textbf{Keywords}: Jump Diffusion; Optimal Dividend; Reinsurance; Stochastic Control

\section{Introduction}
\label{sec:intro}

Stochastic control has enjoyed great success in actuarial science since 1990s; see, e.g., \cite{browne1995optimal} and \cite{asmussen1997controlled} for early contributions and the monograph  \cite{schmidli2007stochastic} 
for a more systematic overview. 
Of particular interest in the application of stochastic control to actuarial problems is the study of optimal dividend strategies for an insurer.
In this paper, we consider an insurer who controls not only dividend payments but also investment and liability strategies simultaneously, and apply a perturbation approach to solve such an optimization problem under the utility maximization criterion. 

The optimal dividend problem is a well-studied topic that, generally speaking, seeks an optimal dividend strategy so as to optimize the insurer's business objective. 
In a seminal paper, \cite{de1957impostazione} models the insurer's risk process $R$ under the classical Crem\'er-Lundberg (CL) setting\footnote{In this setting, the risk process $R=(R_t)_{t \ge 0}$ is modeled by a compound Poisson process $R_t = \sum_{i=1}^{\wh N_t} Y_i$, where $\wh N$ is a homogeneous Poisson process and $(Y_i)_{i=1,2,\cdots}$ is a series of independently and identically distributed random variables, also independent of $\wh N$. Here, $\wh N$ models the claim frequency and $Y_i$ models the severity of the $i$-th claim.} and considers an insurer who seeks to maximize the expected discounted dividend payments until the ruin time $\tau$: 
\begin{align}
	\label{eq:de-obj}
	\max \; \int_0^\tau \, e^{-\delta t} \, \dd \wt D_t,
\end{align}
where $\delta >0$ is the (subjective) discount rate and $\wt D=(\wt D_t)_{t \ge 0}$ denotes the insurer's dividend strategy, with $\wt D_t$ representing the \emph{cumulative} dividend payments up to time $t$. 
Note that the ruin time $\tau$ in \eqref{eq:de-obj} depends on the dividend strategy $\wt D$; let us denote $X=X^{\wt D}$ the insurer's controlled surplus (wealth) process, then $\tau = \tau^{\wt D} := \inf\{t: X^{\wt D} \le 0 \}$.
The optimal dividend strategy $\wt D^*$ to Problem \eqref{eq:de-obj} is often a band strategy and may further reduce to a barrier strategy in certain cases; see  \cite{avanzi2009strategies} for details. 
There is a rich body of literature on the optimal dividend problem by now; see \cite{schmidli2007stochastic}[Sections 2.4 and 2.5] for a standard textbook treatment. 
We refer readers to \cite{albrecher2009optimality} and \cite{avanzi2009strategies} for more comprehensive surveys on this topic. 
In what follows, we conduct a selective literature review on this problem by discussing various extensions to the work of \cite{de1957impostazione} from four different angles: 
the risk model, the type of dividend strategies, additional control components and constraints, and the optimization objective.

First, we discuss extensions to \cite{de1957impostazione} in the modeling of the insurer's risk and (uncontrolled) surplus processes.
When it comes to the modeling of the insurer's risk process $R$, a popular choice is the diffusion model,\footnote{In this setting, the dynamics of the risk process $R=(R_t)_{t \ge 0}$ is given by $\dd R_t = \alpha \dd t + \beta \dd W_t$, where $\alpha, \beta>0$ and $W=(W_t)_{t\ge 0}$ is a standard Brownian motion. Here, the parameters $\alpha$ and $\beta$ can be selected to match the first two moments of $R$ under the CL model; see \cite{browne1995optimal}.} which can be seen as a continuous approximation to the classical CL model.
Please see \cite{jeanblanc1995optimization}, \cite{asmussen1997controlled}, \cite{hojgaard1999controlling}, and \cite{choulli2003diffusion}, among many others for investigations of optimal dividend under the diffusion risk model. 
Further generalizations to the standard diffusion model include incorporating mean-reverting (see \cite{cadenillas2007optimal} and \cite{avanzi2012mean})  and regime switching (see \cite{sotomayor2011classical} and  \cite{jin2013numerical}).
Much of the attention in modeling goes to the risk process $R$, while less is concerned about the premium process $p=(p_t)_{t \ge 0}$. 
In most analyses, $p_t$ is a constant (see, e.g., \cite{taksar1998optimal} and \cite{azcue2005optimal}) or further given by a specific premium principle (see, e.g., \cite{asmussen2000optimal} for the expected value principle). 
But other studies apply a stochastic process (e.g., a diffusion or jump-diffusion process) to model the insurer's surplus without dividend; see, e.g., \cite{choulli2003diffusion}, \cite{gerber2004optimal},  \cite{cadenillas2007optimal}, \cite{chen2014optimal}, and \cite{zhu2015dividend}. 
%Others consider a general premium form that may take into account the time-inhomogeneous feature or the possible dependence between claim frequency and claim severity; see, e.g., \cite{landriault2008constant}.
Opposite to the standard risk model is the so-called dual risk model, where $p$ is now the expense rate and $R$ is interpreted as the positive gains.
Optimal dividend in the dual model is studied in \cite{avanzi2007optimal}, \cite{avanzi2008optimal}, %\cite{yao2011optimal}, 
and \cite{bayraktar2013optimal, bayraktar2014optimal}, among others.

Second, we consider different types of dividend strategies that result in different types of control problems. 
We denote the insurer's dividend strategy by  $\wt D=(\wt D_t)_{t \ge 0}$, where $\wt D_t$ is understood as the cumulative dividend payout up to time $t$.
In the first case, there exists an absolutely continuous and non-negative dividend \emph{rate} $D=(D_t)_{t \ge 0}$ such that $\dd \wt D_t = D_t \dd t$ for all $t \ge 0$; see Case A in \cite{jeanblanc1995optimization}. 
The optimal dividend problem in this case is a classical control problem. 
Often an upper bound on the dividend rate $D$ is imposed (i.e., $0 \le D_t \le m < \infty$) and, in many cases, the optimal dividend strategy is equal to the upper bound if the controlled surplus is large enough, and 0 otherwise. 
The study of optimal dividend in this case can be found in \cite{asmussen1997controlled}, \cite{sotomayor2011classical}, \cite{jin2015optimal}, and \cite{zhu2015dividend}, among others.
In the second case, the intervention (dividend payout) is \emph{not} continuous and a dividend strategy consists of a pair of processes $(T_i, D_i)_{i=1,2,\cdots}$, where $D_i$ is the \emph{amount} of the $i$-th dividend paid at time $T_i$; see Case B in \cite{jeanblanc1995optimization}. 
The control choice in this case is appropriate when paying dividends incurs transaction costs or taxes, and the corresponding problem is an impulse control problem. 
Please refer to \cite{cadenillas2006classical}, \cite{cadenillas2007optimal}, \cite{bai2010optimal}, 
 and \cite{yao2011optimal} 
for related investigations of this case. 
In the third case, a dividend strategy $\wt D$ is a c\`adl\`ag process (right continuous with left limits) that is non-decreasing and non-negative; see Case C in \cite{jeanblanc1995optimization}. 
The optimal dividend problem in this setup often leads to a singular control problem.
\cite{jeanblanc1995optimization} show that when the upper bound $m$ on the dividend rate goes to infinity in the first case or the fixed transaction cost goes to zero in the second case, the corresponding optimal dividend strategy converges to the one in the third case. 
Please consult, e.g., \cite{taksar1998optimal}, \cite{azcue2005optimal}, \cite{AlbrecherHT05}, \cite{bai2010optimal}, and \cite{sotomayor2011classical} for the related study in this case. 

Third, we discuss interesting features, mainly in the control components, that have been incorporated into the original model of \cite{de1957impostazione}.
In \cite{de1957impostazione}, an insurer only controls the dividend payout strategy, and neither invests the surplus in a financial market nor resorts to risk control strategies (e.g., reinsurance policies). 
Much of the following research includes investment and/or risk control strategies into the study of optimal dividend; see, e.g., \cite{taksar1998optimal}, \cite{asmussen2000optimal}, \cite{choulli2003diffusion}, \cite{azcue2005optimal}, and \cite{cadenillas2006classical}  for a short list. 
The insurer's controls and/or actions are often subject to various restrictions. 
An important restriction is transaction costs or taxes on dividend strategies; see \cite{cadenillas2006classical}, \cite{bai2010optimal}, and \cite{schmidli2017capital}. 
Other features/restrictions proposed in the literature include constraints on risk control (see \cite{choulli2003diffusion}), capital injections (see \cite{kulenko2008optimal},  \cite{schmidli2017capital},  \cite{AlbrecherI18}, and \cite{lindensjo2020optimal}), and different credit and debt interest rates  (see \cite{zhu2013optimal}).
The majority of the literature on optimal dividend aggregates the insurer's entire business together in the study, while several recent works consider the insurer as a multi-line business entity, where capital transfers from one business line to another are possible; see \cite{gu2018optimal} and \cite{jin2020optimal}. 

Last, we review the insurer's optimization objectives in the optimal dividend problem. 
The most standard choice is given in \eqref{eq:de-obj}, maximizing the expected discounted dividend payments up to ruin. 
It is known that the ruin probability is 1 under the optimal dividend strategy $\wt D^*$ to Problem \eqref{eq:de-obj}, %i.e., $\Pb(\tau < \infty) =1$ under $\wt D^*$, 
which may be seen as a ``disadvantage for not taking the lifetime of the controlled process into account'' as pointed out in \cite{thonhauser2007dividend}.
Consequently, \cite{thonhauser2007dividend} propose to add a penalty term $-e^{-\delta \tau}$ to the objective function in \eqref{eq:de-obj}; such an idea is similar to the bequest penalty used in  \cite{karatzas1986explicit} to deal with bankruptcy in optimal investment and consumption problems. 
In \eqref{eq:de-obj}, the discounting function is of exponential form, which leads to a time-consistent control problem. 
However, different discounting functions are proposed in the literature that may cause time-inconsistency; see a piece-wise exponential discounting in \cite{chen2014optimal}, a general discounting (within a discrete framework) in \cite{zhou2020optimal}, and quasi-hyperbolic discounting in \cite{zhu2020singular}.
Another important alternative objective is rooted from the optimal investment literature (see \cite{merton1969lifetime, merton1971optimum}) and aims to maximize the expected \emph{utility} of dividend payout; see, e.g., \cite{cadenillas2007optimal}, \cite{grandits2007optimal},  \cite{thonhauser2011optimal}, \cite{jin2015optimal}, and \cite{xu2020optimal}. 
In fact, in an early article \cite{gerber2004optimal}[Sections 9 and 10], the authors point out that ``maximizing the expected utility of the present value of the
dividends until ruin is a new and challenging problem.''

In this paper, we formulate a stochastic control problem for an insurer who chooses a triplet control consisting of investment, liability ratio, and dividend strategies in a combined financial and insurance market. 
We model the insurer's (unit) risk process $R$ 
by a diffusion process modulated with jumps, which is correlated with the risky asset price.
In our framework, the insurer can directly control the liability units as measured by the liability ratio (defined as the ratio of the total liabilities to the controlled surplus). 
The goal of the insurer is to seek an optimal investment, liability ratio, and dividend strategy to maximize her expected utility of dividend payments over an infinite horizon: $\max \, \Eb [ \int_0^\infty \, e^{-\delta t} U(D_t) \, \dd t]$, 
where $\delta >0$ is the discount rate, $D=(D_t)_{t \ge 0}$ denotes the dividend rate, and the utility function $U$ belongs to the hyperbolic absolute risk aversion (HARA) utility family (including log utility and power utility).
%see, e.g., \cite{cadenillas2007optimal} and \cite{thonhauser2011optimal} for a similar objective choice.
To solve such a control problem, we apply a perturbation approach, proposed in \cite{herdegen2020elementary},  that is different from the standard Hamilton-Jacobi-Bellman (HJB) and martingale approaches.
% (see \cite{fleming2006controlled} and \cite{karatzas1998methods}). 
In the first step, we solve the problem for a restricted class of strategies--constant strategies and obtain the optimal constant strategy $u_c^*$ and the corresponding smooth value function $V_c$ in 
closed form.
In the second step, for any admissible strategy, we perturb it by an $\epsilon$-optimal constant strategy\footnote{Such a strategy is the same as the optimal constant strategy $u_c^*$ and starts with an initial wealth $\epsilon>0$.} and show that the value function for the perturbed problem $V_\epsilon(x)$ is equal to $V_c(x + \epsilon)$, and  the optimal constant strategy $u_c^*$ remains optimal among all admissible strategies.

We next summarize the main contributions of this paper as follows. 
In the literature, the HJB approach is the dominate choice to optimal dividend, and, to the best of our knowledge, the perturbation approach developed in \cite{herdegen2020elementary} has not been applied to solve the optimal dividend problem before. 
The perturbation approach is particularly useful in dealing with the HARA utility maximization criterion adopted in this paper and applies simultaneously for both the cases of $0 <\eta \le 1$ and $\eta >1$, where $\eta$ is the relative risk aversion. 
In comparison, the related papers of \cite{cadenillas2007optimal}, \cite{thonhauser2011optimal}, \cite{jin2015optimal}, and \cite{xu2020optimal} only analyze the case of $0 < \eta \le 1$ using the HJB approach. The case of $\eta >1$ often requires additional regularity assumptions and involved analysis under the HJB approach; see Section \ref{sub:comp} for a detailed comparison between our approach and the HJB approach. 
\cite{herdegen2020elementary} consider the Merton problem (optimal consumption/investment) and focus more on the technical side, while we study the optimal dividend problem and aim to obtain explicit solutions and interpret the results from the economic point of view.
The optimal dividend strategy is to pay dividends at a constant proportion of the insurer's wealth, due to the choice of HARA utility, and the corresponding wealth process is always strictly positive. 
Put differently, our optimal dividend strategy provides a stable way of distributing profits to the shareholders and always maintains solvency for the insurer.
In addition, we conduct an extensive economic analysis in Section \ref{sec:econ}, both analytically and numerically, to investigate how various model parameters and the insurer's risk aversion affect the optimal strategies.

The rest of the paper is organized as follows. 
In Section \ref{sec:model}, we introduce the market model and formulate the main stochastic control problem of the paper. 
We first solve the problem for a restricted class of strategies--constant strategies in Section \ref{sec:cont} and next verify the optimal constant strategy is also globally optimal among all the unrestricted admissible strategies in Section \ref{sec:veri}. 
In Section \ref{sec:econ}, we conduct a thorough economic analysis along both analytic and numerical directions to investigate the impact of various parameters on the optimal strategy. 
Finally, we conclude in Section \ref{sec:con}.

\section{Problem Formulation}
\label{sec:model}

\subsection{The Markets}
\label{sub:market}

We consider a representative insurer (``she"), who has access to a financial market consisting of a risk-free asset and a risky asset (e.g., a stock index or a mutual fund).  
The risk-free asset earns at a constant rate $r$ continuously.
The price process $S=(S_t)_{t \ge 0}$ of the risky asset is modeled by 
%$\dd S_t = \mu S_t \, \dd t + \sig S_t \, \dd \wone_t$, 
\begin{align}
\label{eq:dS}
\dd S_t = \mu S_t \, \dd t + \sig S_t \, \dd \wone_t,
\end{align}
where $\mu, \sig>0$ and $\wone$ is a standard Brownian motion. 
We assume the underlying financial market is ideal and frictionless.

The insurer's main business is to underwrite policies against insurable risks. Following \cite{zou2014optimal}, we model such risks on a unit basis (e.g., per policy) by 
\begin{align}
\label{eq:dR}
	\dd R_t = \alpha \, \dd t + \beta \rho \, \dd \wone_t + \beta \sqrt{1 - \rho^2} \, \dd \wtwo_t + \gamma \, \dd N_t,
\end{align}
where $\alpha, \beta, \gamma >0$, $-1 < \rho < 1$, $\wtwo$ is another standard Brownian motion, and $N$ is a homogeneous Poisson process with constant intensity $\lambda>0$. 
Let $p$ represent the \emph{unit} premium rate, corresponding to the unit risk process $R$ in \eqref{eq:dR}. 
We further impose the following assumptions on the model \eqref{eq:dS}-\eqref{eq:dR} throughout the rest of this paper:
\begin{align}
	\label{eq:assu}
	\mu > r  \qquad \text{and} \qquad p > \alpha + \lam \gamma.
	%, \qquad \sig, \alpha, \beta, \gamma >0, \qquad -1 < \rho < 1.
\end{align}
We discuss the economic interpretations of the above assumptions in \eqref{eq:assu} in Remark \ref{rem:model}.
On the technical level, we assume the processes $\wone$, $\wtwo$, and $N$ are independent under a complete probability space $(\Omega, \Fc, \Pb)$ and the filtration $\Fb:=(\Fc_t)_{t \ge 0}$ is generated by these three processes, augmented with $\Pb$-null sets. 

\begin{remark}
	\label{rem:model}
In the above described financial market, we model the risky asset by the classical Black-Scholes model \eqref{eq:dS}, which is a dominant choice in the study of optimal investment problems due to trackability; see \cite{merton1969lifetime, merton1971optimum}. 
In the literature, a standard approach for modeling the risk process is the Crem\'er-Lundberg (CL) model, while an equally popular alternative is the diffusion model, which can be seen as a continuous approximation of the CL model (a compound Poisson process). Please refer to \cite{browne1995optimal} and \cite{hojgaard1998optimal} for early applications of the diffusion risk model in actuarial science. Here our risk model \eqref{eq:dR} is a further generalization of both the CL model and the diffusion model.
Indeed, if we set $\alpha = \beta =0$ and allow $\gam$ to be a random variable, then \eqref{eq:dR} becomes the classical CL model.
On the other hand,  by setting $\rho = \gam = 0$ in \eqref{eq:dR}, our model is reduced to the standard diffusion model considered in \cite{browne1995optimal} and \cite{asmussen1997controlled}. 
Our first generalization in \eqref{eq:dR} is to allow the risk process from insurance underwriting to be correlated with the risky asset in the financial market.	
As argued in \cite{stein2012stochastic}[Chapter 6], a major mistake in the AIG's business operation during the financial crisis of 2007-2008 is ignore or underestimate the \emph{negative} correlation between the financial market and its insurance liabilities. Such a negative correlation can be easily captured by our risk model  \eqref{eq:dR} by setting $\rho \in (-1, 0)$; see, e.g., \cite{zou2014optimal} and \cite{jin2015optimal}. 
Second, we incorporate jumps into the modeling of the risk process $R$. Here jumps can help capture sudden extreme claims, which may not be modeled by a continuous process such as a Brownian motion. We refer readers to \cite{shen2020mean} for more detailed discussions on the risk model \eqref{eq:dR}.

The two assumptions in \eqref{eq:assu} both have reasonable economic meanings. 
The first condition implies that the Sharpe ratio of the risky asset, $\Lam := \frac{\mu - r}{\sig}$, is strictly positive, which holds true for all ``good" assets (e.g., stock indexes) over long run in the financial market.
Since $\Eb[\dd R_t] = (\alpha + \lam \gamma) \dd t$, the second condition implies that the insurer should charge the unit premium rate $p$ greater than $\alpha + \lam \gamma$, the so-called ``actuarial fair price".
Otherwise (i.e., if $p \le \alpha + \lam \gamma$),  the insurer's ruin probability  (without investment) is equal to 1. 
If we apply the expected value principle to calculate insurance premium (i.e., $p = (1+\theta) \, (\alpha + \lam \gamma)$), then assuming $p > \alpha + \lam \gam$ is equivalent to setting $\theta>0$, where $\theta$ is often called the loading factor. 
Note that we do \emph{not} assume that $p$ is given by the above expected value principle, which is nevertheless imposed in many works in order to obtain explicit solutions; see, e.g., \cite{asmussen2000optimal}.
\end{remark}

\subsection{The Insurer's Strategies and Wealth Process}

The representative insurer chooses a triplet strategy (control) that consists of an investment strategy, a liability ratio strategy, and a dividend strategy in the business operations, as described in what follows.

\begin{itemize}
	\item In the financial market, the insurer chooses a dynamic investment strategy $\pi = (\pi_t)_{t \ge 0}$, where $\pi_t$ denotes the \emph{proportion} of wealth invested in the risky asset at time $t$. 
	
	\item In the insurance market, we assume the insurer can directly control the amount of liabilities (number of units or policies) in the underwriting, denoted by $L = (L_t)_{t \ge 0}$. 
	Following \cite{stein2012stochastic}, we define  $\kappa_t := L_t / X_t$, where $X_t$ is the insurer's wealth at time $t$ (defined later in \eqref{eq:dX}), and call $\kappa = (\kappa_t)_{t \ge 0}$ the insurer's \emph{liability ratio} strategy.
	
	\item The insurer chooses a dividend strategy  $D = (D_t)_{t \ge 0}$ to distribute profits to the shareholders, where $D_t$ denotes the continuous dividend \emph{rate} payable at time $t$. 
	Namely, the dividend payment over $[t, t + \dd t]$ is given by $D_t \dd t$.  
\end{itemize}
We denote the insurer's strategy by $u := (\pi, \kappa, D)$. 
We offer some explanations to the insurer's liability strategy $\kappa$ and dividend strategy $D$ in the following remark.

\begin{remark}
	\label{rem:control}
Motivated by the AIG case during the financial crisis of 2007-2008, \cite{stein2012stochastic} finds that 
the liability ratio provides an early warning signal to the AIG's failure and sets up a model in which the insurer directly controls its liability ratio.
Such a setup attracts considerable attention in the actuarial science literature, as evidenced by a series of follow-up works on the optimal control study for an insurer; see \cite{zou2014optimal}, \cite{jin2015optimal},  and \cite{shen2020mean}, among many others. 
From an economic point of view, the assumption that the insurer can dynamically control the liability ratio is not unreasonable, as major insurers have monopoly power in the insurance market and may also  ``discriminate" policy holders.
For instance, in the business of health insurance, insurers often reject potential policy applications based on certain risk factors. 
\cite{lapham1996genetic} study a group of participants with genetic disorders in the family and find that 25\% of them believed they were refused life insurance and 22\% were refused health insurance. 
A recent work of \cite{bernard2020optimal} confirms that, under certain circumstances, ``it may become optimal for the insurer to refuse to sell insurance to some prospects". 

The dividend strategy $D$ described above is absolutely continuous with respect to the Lebesgue measure, which is used in many related works; see \cite{asmussen1997controlled},  \cite{avanzi2012mean}, and \cite{jin2015optimal} for a short list. 
As such, the corresponding control problem is a classical one, instead of a singular or impulse one.  
We refer to \cite{jeanblanc1995optimization} and \cite{sotomayor2011classical} for both classical and singular control of optimal dividend problems. 
As argued in \cite{avanzi2012mean}, optimal dividend strategies obtained in the literature are often volatile (e.g., ``bang-bang'' strategies), which are unlikely to be adopted by managers.
In consequence, they consider a setup where dividends are paid at a constant rate $g$ of the company's surplus and seek to find the optimal rate $g^*$. 
Note that the above dividend strategy $D$ includes the threshold strategies (also called refracting strategies), where $D_t = \text{constant}$ when the wealth (surplus) at time $t$ is greater than a given threshold and $D_t = 0$ otherwise.
For analyses on the threshold-type dividend strategies, please see  \cite{albrecher2009optimality} and \cite{avanzi2009strategies}.
%, and the recent article of \cite{albrecher2018dividends}.
% and \cite{renaud2020stochastic}.
\end{remark}

Let us denote by $X=(X_t)_{t \ge 0}$ the insurer's wealth  process (also called surplus process) under a triplet strategy of investment, liability ratio, and dividend  $u := (\pi, \kappa, D)$. 
The dynamics of $X$ is obtained by 

\begin{align}
\label{eq:dX}
\dd X_t &= \big( X_{t-} (r + (\mu - r) \pi_t + (p - \alpha) \kappa_t) - D_t \big) \, \dd t 
+ (\sig \pi_t - \beta \rho \kappa_t) X_{t-} \, \dd \wone_t \\
&\quad - \beta \sqrt{1 - \rho^2} \, \kappa_t X_{t-} \, \dd \wtwo_t - \gam \kappa_t X_{t-} \, \dd N_t, \qquad   X_0 = x > 0.
\end{align}
It is clear from \eqref{eq:dX} that the insurer's wealth process depends on her initial wealth $x$ and strategy $u$, and we write it as $X$ instead of $X^{x,u}$ for 
%national 
notational simplicity. 

\begin{remark}
We explain in this remark how we derive the dynamics of the insurer's wealth process $X$ in \eqref{eq:dX}. 
Let us decompose $\dd X_t$ into three components: $\dd X_t = \Delta_t^i + \Delta_t^l + \Delta_t^d$, where $\Delta_t^i$ (resp. $\Delta_t^l$ and $\Delta_t^d$) denotes the wealth changes due to investment (resp. insurance liabilities and dividend payments).	
Given an investment strategy $\pi$, the amount of $\pi X$ is invested in the risky asset while the remaining $(1 - \pi) X$ is invested in the risk-free asset, implying $\Delta_t^i = \frac{\pi_t X_{t-}}{S_t} \, \dd S_t + r (1 - \pi_t) X_{t-} \, \dd t = X_{t-} (r + (\mu - r)\pi_t) \, \dd t + \sigma \pi_t X_{t-} \, \dd W_t^{(1)}$.
Given a liability ratio strategy $\kappa$, the insurer's total liabilities $L$ from underwriting insurance policies is given by $L = \kappa X$. Recall from Section \ref{sub:market} that the risk $R$ and the premium rate per unit liabilities are given respectively by \eqref{eq:dR} and $p$. As such, we have $\Delta_t^l = L_t (p \, \dd t - \dd R_t) = \kappa_t X_{t-} \big[ (p - \alpha)  \, \dd t - \beta \rho \, \dd W_t^{(1)} - \beta \sqrt{1-\rho^2} \, \dd W_t^{(2)} - \gamma \, \dd N_t \big] $. For a given dividend strategy $D$, it is easily seen that $\Delta_t^d = - D_t \, \dd t$. Therefore, upon combining the above results, we obtain the dynamics $\dd X$ as in \eqref{eq:dX}.
\end{remark}

We next define the admissible strategies in Definition \ref{def:ad} and close this subsection with some remarks. 

\begin{definition}[Admissible Strategies]
	\label{def:ad}
	A strategy $u=(\pi, \kappa, D)$ is called admissible, denoted by $u \in \Ac$, if (1) $u$ is predictable with respect to the filtration $\Fb$; 
	(2) for all $t \ge 0$, we have $\int_0^t \, \pi_s^2 \, \dd s < \infty$, $0 \le \kappa_t < \frac{1}{\gam}$, and $\Eb[\int_0^t \, D_s \, \dd s]<\infty$ with $D_t \ge 0$;
	(3) there exists a unique strong solution to \eqref{eq:dX} such that $X_t \ge 0$ for all $t \ge 0$.
	%; and (4) the functional $J$ in \eqref{eq:J} is well defined.
\end{definition}

\begin{remark}
\label{rem:ad}
The constraint $0 \le \kappa_t < \frac{1}{\gam}$ stipulates that the wealth process $X$ will not become negative or zero at the jump times of the Poisson process $N$. 
We also mention that the condition on the investment strategy $\pi$ is weaker than the standard square integrability condition $\Eb [\int_0^t \, \pi_s^2 \, \dd s ] < \infty$ for all $t \ge 0$.
\end{remark}

\subsection{The Problem}

Given a triplet strategy $u$, we define the insurer's objective functional $J$ by 
%(see \cite{jin2015optimal})
\begin{align}
\label{eq:J}
J(x; u) := \Eb \left[\int_0^\infty \, e^{-\delta t} \, U(D_t) \, \dd t\right],
\end{align}
where $\Eb$ denotes taking expectation under the physical measure $\Pb$, $\delta>0$ is the subjective discount factor, and $U$ is a standard utility function (increasing and concave).  
The concavity of $U$ implies risk aversion in the insurer's decision making.
Applying a concave utility function to the dividend payments is in fact not unusual in the optimal dividend literature. 
\cite{cadenillas2007optimal} provide economic justifications for incorporating risk aversion in the objective (see pp.85 therein); see also \cite{thonhauser2011optimal}, \cite{jin2015optimal}, and \cite{xu2020optimal}. 

\begin{remark}
We add a technical remark to address the possible ruin and its impact on the objective $J$ in \eqref{eq:J}.
To that end, for any admissible strategy $u \in \Ac$, define the ruin time $\tau^u$ by 
\begin{align}
	\tau^u := \inf\{t \ge 0: \, X \le 0\}, \qquad u \in \Ac,
\end{align}
where the insurer's wealth process $X$ is defined in \eqref{eq:dX}. 
After ruin, we set the utility derived from dividends to be a negative constant $P$, i.e., $U(D_t) \equiv P$ when $t \ge \tau^u$. 
Note that given a utility function $U$ (which is increasing and concave), $(U+ \, \text{constant})$ remains a utility function. 
Therefore, the introduction of $P$ for the utility after the ruin time $\tau^u$ does not destroy monotonicity or concavity (upon a proper shift).

We then rewrite $J$ in \eqref{eq:J} as 
\begin{align}
	\label{eq:J-new}
	J = \Eb \left[\int_0^{\tau^u} \, e^{-\delta t} \, U(D_t) \, \dd t + \int_{\tau^u}^\infty \, e^{-\delta t} \, U(D_t) \, \dd t\right] = \Eb \left[\int_0^{\tau^u} \, e^{-\delta t} \, U(D_t) \, \dd t + \frac{P}{\delta} \, e^{-\delta \tau^u} \right].
\end{align}
The negative constant $P$ in \eqref{eq:J-new} can be interpreted as a free parameter penalizing early ruin. 
The same penalty term is also adopted in \cite{thonhauser2007dividend} in order to take into account the lifetime of the controlled process; see also \cite{karatzas1986explicit} for a similar idea. 
If we let $P \to 0$, the objective in \eqref{eq:J-new} reduces to \eqref{eq:de-obj}, the standard choice in the literature; see \cite{albrecher2009optimality} and \cite{avanzi2009strategies}. 
In other words, early ruin is \emph{not} penalized in the \cite{de1957impostazione} model. 
If we let $P \to -\infty$, i.e., a finite ruin is penalized infinitely, then the insurer will adopt a strategy that does not result in ruin ($\tau^u = +\infty$). 
Our later results show that the insurer's wealth process is strictly positive under the optimal strategy, hence the choice of $P$ is irrelevant in this paper. The purpose of rewriting $J$ in \eqref{eq:J-new} is to relate our objective functional $J$ to the classical one given by \eqref{eq:de-obj}.
\end{remark}

We now formulate the main optimal investment, liability ratio, and dividend problem of this paper.

\begin{problem}
	\label{prob:main}
The insurer seeks an optimal strategy $u^*=(\pi^*, \kappa^*, D^*)$
%=(\pi^*, \kappa^*, D^*)$ 
to maximize the expected discounted utility of dividend over an infinite horizon. 
Equivalently, the insurer solves the following stochastic control problem:
\begin{align}
	\label{eq:prob}
	V(x) := \sup_{u \in \Ac} \; J(x; u) = J(x; u^*),
\end{align}
where the admissible set $\Ac$ is defined in Definition \ref{def:ad} and the objective functional $J$ is defined in \eqref{eq:J}.
\end{problem}

In this work, we consider the utility function $U$  of  the hyperbolic absolute risk aversion (HARA) family. In particular, we assume $U$ is given by 
\begin{align}
\label{eq:U}
%U(x) = \ln(x) \qquad \text{or} \qquad 
U(x) = \frac{1}{1 - \eta} \, x^{1 - \eta},  
%\quad 0 <\eta <1 \text{ or } \eta > 1,
\qquad 
\eta > 0,
\end{align}
where the limit case of $\eta = 1$ is treated as log utility $U(x) = \ln x$.
Note that $U$ in \eqref{eq:U} is well defined for all $x>0$. 
If $U$ is given by log utility ($\eta  = 1$) or negative power utility ($\eta > 1$), we set $U(0) = -\infty$. 
The choice of power utility is dominant in the optimal investment literature (see the classical papers of \cite{merton1969lifetime, merton1971optimum}).
We comment that the relative risk aversion $\eta$ in \eqref{eq:U} can be any positive number, and is in particular allowed to be greater than 1, arguably the case in real life; see \cite{meyer2005relative}.
In comparison,  $ 0 < \eta \le 1$ is assumed in \cite{cadenillas2007optimal}, \cite{thonhauser2011optimal}, \cite{jin2015optimal}, and \cite{xu2020optimal}.

\section{Analysis of Constant Strategies}
\label{sec:cont}

In this section, we study Problem \eqref{eq:prob} over a restricted set of \emph{constant} strategies $\Ac_c$, which is defined by
\begin{align}
\label{eq:ad_cont}
\Ac_c := \{ u = (\pi, \kappa, D) \, | \, \pi_t \equiv \pi_c, \, \kappa_t \equiv \kappa_c, \, D_t \equiv \xi_c X_{t}, \; \text{ where } \pi_c, \kappa_c, \xi_c \in \Rb, \, \kappa_c \in [0, 1/\gam), \, \xi_c \ge 0 \}.
\end{align}
We denote constant strategies in $\Ac_c$ by $u_c :=(\pi_c, \kappa_c, \xi_c)$, which is slightly different from $u = (\pi, \kappa, D) \in \Ac$ introduced in Section \ref{sec:intro}.
For any $u_c\in \Ac_c$, there exists a unique strong (positive) solution $X$ to
% the stochastic differential equation (SDE)  
\eqref{eq:dX} and $\Eb[X_t^2] < \infty$ for all $t \ge 0$; see, e.g., Theorem 1.19 in \cite{oksendal2005applied}. 
This result, along with the definition in \eqref{eq:ad_cont}, implies that all the conditions in Definition \ref{def:ad} are satisfied.
%, and the fourth condition holds under some technical assumptions on the model parameters (to be imposed in Proposition \ref{prop:log} or \ref{prop:pow_cont}).  
Therefore, we conclude that the set of constant strategies $\Ac_c$ is a (proper) subset of the set of admissible strategies $\Ac$.

To begin, we solve \eqref{eq:dX} explicitly and  obtain the insurer's wealth $X_t$ at time $t$ (for all $t \ge 0$)  by
\begin{align}
\label{eq:X_cont}
X_t = x \, \exp\left( \big({f}(\pi_c, \kappa_c) - \xi_c \big)t + \big(\sig \pi_c - \beta \rho \kappa_c \big) \wone_t - \beta \sqrt{1 - \rho^2} \kappa_c \wtwo_t + \ln(1 - \gam \kappa_c) \, \wt{N}_t  \right),
\end{align}
where $\wt N = (\wt N_t)_{t \ge 0}$,  with $\wt{N}_t := N_t - \lam t$, is the compensated Poisson process. In addition, the function $f$ is  defined over $\Rb \times [0, 1/\gam)$ by 
\begin{align}
\label{eq:f}
f(y_1, y_2) := r + (\mu - r) y_1 + (p - \alpha) y_2 - \frac{1}{2} \big( \sig^2 y_1^2
- 2 \beta \rho \sig y_1 y_2 +  \beta^2 y_2^2 \big) + \lam \ln(1 - \gam y_2).
\end{align}
For future convenience, we define three constants $A$, $B$, and $C$ by 
\begin{align}
	\label{eq:ABC}
	A:=  \gam \beta^2 ( 1- \rho^2), \qquad B:=  \beta^2(1-\rho^2) + \gam (p - \alpha + \beta \rho \Lam), \qquad 
	C:= p - \alpha + \beta \rho \Lam - \lam \gam,
\end{align}
where $\Lam$ is the Sharpe ratio of the risky asset, i.e., 
\begin{align}
\label{eq:Lam}
\Lam := \frac{\mu - r}{\sig}.
\end{align}
Due to $\rho \in (-1,1)$ and \eqref{eq:assu}, we have $A > 0$ and $B > 0$. 
If $\rho \ge 0$, then $C > 0$; but if $\rho < 0$, $C$ may be negative.

We first analyze the case of log utility $U(x) = \ln x$, corresponding to $\eta = 1$ in \eqref{eq:U}. 
Note that for an admissible constant strategy $u_c \in \Ac_c$, the corresponding wealth $X_t > 0$ for all $X(0)=x>0$ as seen from \eqref{eq:X_cont}. 
Using \eqref{eq:ad_cont}-\eqref{eq:X_cont} along with the definition of $J$ in \eqref{eq:J}, we obtain
\begin{align}
	\label{eq:log_1}
	\Eb \left[\int_0^\infty \, e^{-\delta t} U(D_t) \, \dd t \right] 
	%= \int_0^\infty \, e^{-\delta t} \big(\ln x + f(\pi_c, \kappa_c) t + \ln \xi_c - \xi_c t\big) \dd t
	=\frac{1}{\del} \ln x + \frac{1}{\del} \ln \xi_c + \frac{1}{\del^2} f(\pi_c, \kappa_c)  - \frac{1}{\del^2} \xi_c.
\end{align}
Solving Problem \eqref{eq:prob} over $\Ac_c$ when $U(x) = \ln x$ is now equivalent to maximizing the right hand side of \eqref{eq:log_1}, which is solved in the proposition below. 

\begin{proposition}
	\label{prop:log}
	Suppose $U(x) = \ln x$ and the constant $C$ defined in \eqref{eq:ABC} is non-negative. 
	The optimal constant strategy $u_c^*= (\pi_c^*, \kappa_c^*, \xi_c^*)$ to Problem \eqref{eq:prob} over $\Ac_c$ is given by 
	\begin{align}
		\label{eq:log_cont}
	\pi_c^*= \frac{\mu - r}{\sig^2} + \frac{\rho \beta}{\sig} \, \kappa_c^*, \qquad 
	\kappa_c^* = \frac{B - \sqrt{B^2 - 4AC}}{2A}, \qquad 
	\xi_c^* = \delta,
	\end{align}
where the constants $A$, $B$, and $C$ are defined in \eqref{eq:ABC}.
Furthermore,  the value function $V_c$ is obtained by
	\begin{align}
		\label{eq:log_cont_value}
		V_c(x) := \sup_{u_c \in \Ac_c} \Eb \left[\int_0^\infty \, e^{-\delta t} \ln(D_t) \, \dd t \right] = J(x; u_c^*) = \frac{1}{\delta} \ln (\delta x) + \frac{1}{\delta^2} (\mathfrak{f}^* - \delta),
	\end{align} 
	where $\mathfrak{f}^* :=f(\pi_c^*, \kappa_c^*) = \max %\limits_{(y_1, y_2) \in \Rb \times [0, 1/\gam)}
	 \, f(y_1, y_2)$, with $f$ defined in \eqref{eq:f}, and $\pi_c^*$ and $\kappa_c^*$ derived in \eqref{eq:log_cont}.
\end{proposition} 

\begin{proof}
By applying the first-order condition to maximizing the right hand side of \eqref{eq:log_1}, we obtain that $\pi_c^*$ and $\xi_c^*$ are given by \eqref{eq:log_cont} and $\kappa_c^*$ solves a quadratic equation $A \, y^2 - B \, y + C = 0$, where the constants $A$, $B$, and $C$ are defined in \eqref{eq:ABC}. We compute $B^2 - 4AC= \big( \beta^2 ( 1-\rho^2) - \gam (\lam \gam + C)\big)^2 + 4 \lam B^2 \gam^2 (1-\rho^2) >0$, implying that there are two candidate solutions to the quadratic equation. 
We next verify that the bigger solution is greater than $1/\gam$ and the smaller solution is less than $1/\gam$. 
Consequently, $\kappa_c^*$ is given by the smaller solution as shown in \eqref{eq:log_cont}, which is non-negative if and only if $C \ge 0$.
By verifying the second-order condition, we confirm that $u_c^*$ in \eqref{eq:log_cont} is indeed an optimal constant strategy over $\Ac_c$ to Problem \eqref{eq:prob}. Plugging $u_c^*$ back into $J$ in \eqref{eq:J}, after tedious computations, leads to the value function in \eqref{eq:log_cont_value}.
\end{proof}

We next study the case of power utility $U(x) = x^{1-\eta} / (1 - \eta)$, where $\eta >0$ and $\eta \neq 1$. For any $u_c \in \Ac_c$, we obtain 
\begin{align}
\label{eq:pow_J_cont}
	J(x; u_c) = \Eb \left[\int_0^\infty \, e^{-\delta t} \frac{(\xi_c X_t)^{1-\eta}}{1-\eta} \, \dd t \right] = \frac{x^{1-\eta}}{1 - \eta} \, \frac{\xi_c^{1-\eta}}{\delta - (1-\eta) \cdot g(\pi_c, \kappa_c) +(1-\eta) \xi_c },
\end{align}
provided $\delta - (1-\eta) \cdot g(\pi_c, \kappa_c) +(1-\eta) \xi_c  > 0$.
Here, the function $g$ is defined over $\Rb \times [0, 1/\gam)$ by 
\begin{align}
\label{eq:g} \qquad 
g(y_1, y_2) :=  r + (\mu - r) y_1 + (p - \alpha) y_2 - \frac{\eta}{2} \big( \sig^2 y_1^2 - 2 \beta \rho \sig y_1 y_2 +  \beta^2 y_2^2 \big)  + \frac{\lam}{1 - \eta} \left((1-\gam y_2)^{1-\eta} - 1\right).
\end{align}

We now present the main result for power utility as follows.

\begin{proposition}
\label{prop:pow_cont}
Suppose $U(x) = x^{1 - \eta}/(1-\eta)$, where $\eta > 0$ and  $\eta \neq 1$, and the following two conditions hold
\begin{align}
\label{eq:psi}
C \ge 0 \qquad \text{and} \qquad \psi := \frac{\delta - (1-\eta) \mathfrak{g}^*}{\eta} > 0,
\end{align}
where $C$ is defined in \eqref{eq:ABC} and $\mathfrak{g}^*: 
%=g(\pi_c^*, \kappa_c^*)
= \max %\limits_{(y_1, y_2) \in \Rb \times [0, 1/\gam)} \, 
g(y_1, y_2)$, with $g$ defined in \eqref{eq:g}.
%, and $\pi_c^*$ and $\kappa_c^*$ given by \eqref{eq:pow_cont_op}.
The optimal constant strategy $u_c^*= (\pi_c^*, \kappa_c^*, \xi_c^*)$ to Problem \eqref{eq:prob} over $\Ac_c$ is given by
\begin{align}
	\label{eq:pow_cont_op} 
	\pi_c^* = \frac{\mu - r}{\eta \sig^2} + \frac{\rho \beta}{\sig} \, \kappa_c^*, \qquad 
	\kappa_c^* = \text{unique solution of \eqref{eq:pow_cont_ka}}, \qquad
	\xi_c^* = \psi, 
\end{align}
where $\kappa_c^*$ solves the following non-linear equation 
\begin{align}
\label{eq:pow_cont_ka}
(1 - \gam y)^{-\eta} + \frac{\eta A}{\lam \gam^2} \, y - \frac{C}{\lam \gam} - 1 = 0.
\end{align}
%with the constants $A$ and $C$ defined in \eqref{eq:ABC}. 
Furthermore, the value function $V_c$ is obtained by 
\begin{align}
\label{eq:V_pow}
V_c(x) :=\sup_{u_c \in \Ac_c} \Eb \left[\int_0^\infty \, e^{-\delta t} \, \frac{D_t^{1-\eta}}{1 - \eta} \, \dd t \right] = J(x; u_c^*) = \frac{\psi^{-\eta}}{1 - \eta} x^{1- \eta}.
\end{align}
\end{proposition}

\begin{proof}
Applying the first-order condition to the maximizing problem of the right hand side of \eqref{eq:pow_J_cont} leads to the results of $u_c^*$ in \eqref{eq:pow_cont_op} and the non-linear equation of $\kappa_c^*$ in \eqref{eq:pow_cont_ka}, which has a unique solution in $[0, 1/\gam)$ due to the condition $C \ge 0$. By imposing $\psi >0$, we obtain $\xi_c^* = \psi >0$ and $\delta - (1-\eta) g(\pi_c, \kappa_c) >0$ for all $(\pi_c, \kappa_c) \in \Rb \times [0, 1/\gam)$. 
A straightforward verification process then completes the proof.
\end{proof}

By \eqref{eq:psi}, we have $\lim_{\eta \to 1} \, \psi = \delta$.  
We also notice that, when $\eta \to 1$, the non-linear equation \eqref{eq:pow_cont_ka} reduces to the quadratic equation satisfied by $\kappa_c^*$ in the log utility case.
%$A \, y^2 + B \, y + C = 0$, which leads to the optimal liability strategy in \eqref{eq:log_cont}. 
As such, Proposition \ref{prop:log} for log utility can be seen as the limit result of Proposition \ref{prop:pow_cont} for power utility.
In the literature, a diffusion process \emph{without} jumps is commonly used to approximate the classical Cra\'mer-Lundberg model; see, e.g., \cite{browne1995optimal} and \cite{hojgaard1998optimal}. 
Such a setup corresponds to setting $\lam = 0$ in the risk process \eqref{eq:dR}. 
The results in Proposition \ref{prop:pow_cont} can be simplified when $\lam = 0$, as shown in the corollary below.

\begin{corollary}
\label{cor:power}
Suppose the utility function $U$ is given by \eqref{eq:U} and
there are no jumps ($\lam =0$) in the risk process \eqref{eq:dR}. 
Further suppose two technical conditions hold:
\begin{align}
\label{eq:cor_cond}
p - \alpha + \beta \rho \Lam>0 \qquad \text{ and } \qquad \delta > (1-\eta) \hat{\mathfrak{g}}^*, 
\end{align}
where the constant $\hat{\mathfrak{g}}^*$ is defined by 
\begin{align}
	\label{eq:g-hat}
\hat{\mathfrak{g}}^*:= r + \frac{(p - \alpha + \beta \rho \Lam)^2}{2 \eta \beta^2 ( 1 - \rho^2)} + \frac{\Lam^2}{2 \eta}.
\end{align}
The optimal constant strategy $u_c^*= (\pi_c^*, \kappa_c^*, \xi_c^*)$ to Problem \eqref{eq:prob} over $\Ac_c$ is given by
\begin{align}
	\label{eq:cor_op} 
	\pi_c^* = \frac{\mu - r}{\eta \sig^2} + \frac{\rho \beta}{\sig} \, \kappa_c^*, \qquad 
	\kappa_c^* = \frac{p - \alpha + \beta \rho \Lam}{\eta \beta^2 (1 - \rho^2)}, \qquad
	\xi_c^* = \frac{\delta - (1-\eta) \hat{\mathfrak{g}}^*}{\eta}.
\end{align}
\end{corollary}

\begin{proof}
Introduce $\hat{g} := g|_{\lam = 0}$, where $g$ is defined in \eqref{eq:g}. Maximizing $\hat{g}$ and \eqref{eq:pow_J_cont} leads to the above optimal strategy \eqref{eq:cor_op}.  
Imposing the conditions in \eqref{eq:cor_cond} guarantees that $\kappa_c^*>0$ and $\xi_c^* >0$ in \eqref{eq:cor_op}.
Note that $\hat{\mathfrak{g}}^* = \max\, \hat{g}(y_1,y_2) = \hat g(\pi_c^*, \kappa_c^*)$, where $\pi_c^*$ and $\kappa_c^*$ are obtained in \eqref{eq:cor_op}. 
\end{proof}

We end this section by offering some explanations on the technical assumptions $C \ge 0$ and $\psi > 0$ imposed in Proposition \ref{prop:pow_cont} (or $C \ge 0$  in Proposition \ref{prop:log}). 
First, since $\lim_{\eta \to 1} \, \psi = \delta > 0$ automatically holds, $\psi > 0$ is not needed for the log utility case in Proposition \ref{prop:log}. 
By the model assumption in \eqref{eq:assu} and the definition of $C$ in \eqref{eq:ABC}, 
if $\rho \ge 0$ (recall $\rho$ is the correlation coefficient between the risky asset and the risk process), we always have $C \ge 0$ and such an assumption becomes redundant for $\rho \ge 0$.
If $\rho < 0$ is the true scenario (as argued in \cite{stein2012stochastic}), 
the assumption of $C \ge 0$ means that the risky investment opportunity cannot be too ``good", comparing to the insurance business.
To see this, we rewrite $C < 0$ as $\Lam = \frac{\mu - r}{\sig} > \frac{p - \alpha - \lam \gam}{-\rho \beta}$, where the left hand side is the Sharpe ratio of the risky asset, and the numerator of the right hand side measures the insurer's expected profit (including operation costs) from the insurance business. 
In practice, on the one hand, insurers are only allowed to invest in ``safe" risky assets (with relatively low Sharpe ratio); on the other hand, many insurance businesses are lucrative and insurers charge sufficient safe loading in premiums for solvency and profitability reasons. 
Therefore, imposing $C \ge 0$ not only guarantees the non-negativity of $\kappa_c^*$ but also makes economic sense. 
The other technical assumption $\psi > 0$ makes Problem \eqref{eq:prob} well-posed. 
Indeed, if this assumption fails (i.e., $\psi \le 0$), we have $V(x) = + \infty$ when $0 < \eta < 1$ and $V(x) = - \infty$ when $\eta > 1$ (see Corollary 6.5 in \cite{herdegen2020elementary}).  
By \eqref{eq:psi}, $\psi > 0 \Leftrightarrow \delta > \eta + (1-\eta) \mathfrak{g}^*$, i.e., the subjective discount rate $\delta$ should be greater than a threshold. Similar conditions are commonly imposed for infinite-horizon control problems; see, e.g., \cite{jin2015optimal}[Eq.(4.11)].

\section{Verification for Admissible Strategies}
\label{sec:veri}

We consider the HARA type utility function $U$, given by \eqref{eq:U}, in the formulation of our main stochastic problem (see Problem \eqref{eq:prob}). 
The special \emph{scale property} of the value function inherited from the HARA utility and the classical results from optimal investment problems (see \cite{merton1969lifetime,merton1971optimum}) motivate us to make the following conjecture:
\begin{quote}
	The optimal strategy $u^*$ over the admissible set $\Ac$ to Problem \eqref{eq:prob} is a constant strategy, and hence coincides with the optimal constant strategy $u_c^*$ obtained in Section \ref{sec:cont}.
\end{quote}  
The goal of this section is to verify that the above conjecture is indeed correct when $U$ is given by  \eqref{eq:U}.

\subsection{Notations} 
\label{sub:notation}

Previously in Sections \ref{sec:model} and \ref{sec:cont}, we have used simplified notations to make presentation more concise, since there is no risk of confusion there. 
Now we need to introduce notations in a more rigorous way for the general analysis.
Given an initial wealth $x >0$ and an admissible control $u$, we denote the insurer's wealth at time $t$  by $X_t^{x, u}$, for all $t \ge 0$, which satisfies the stochastic differential equation (SDE) in  \eqref{eq:dX}. 
By Definition \ref{def:ad}, the set of admissible strategies $\Ac$ depends on the insurer's initial wealth $x$ ($x>0$), and we will write it as $\Ac(x)$. 
However, it is still safe to use $\Ac_c$ to denote the set of constant strategies, as its definition in \eqref{eq:ad_cont} is independent of the initial wealth.
Introduce $\xone_t := X_t^{1, u_c^*}$ and $\dione_t:= \xi_c^* \xone_t$ for all $t \ge 0$, where $u_c^*$ is the optimal constant strategy obtained in \eqref{eq:pow_cont_op}. 
Namely, $\xone$ (resp. $\dione$) is the corresponding wealth process (resp. dividend process) under the unit initial wealth ($x=1$) and  the optimal constant strategy $u_c^*$.

Let us take two arbitrary admissible controls $u_1 = (\pi_1, \kappa_1, D_1) \in \Ac(x_1)$ and $u_2 = (\pi_2, \kappa_2, D_2) \in \Ac(x_2)$, where 
$x_1, x_2 >0$. 
We define a new control $u=(\pi, \kappa, D)$, denoted by $u := u_1 \oplus u_2$, along with the corresponding wealth $X^{x,u} = ({X}_t^{ x, u})_{t \ge 0}$ such that the following holds true: 
$x = x_1 + x_2$, $\pi_t {X}_t^{x, u} = \pi_{1,t} {X}_{t}^{x_1, u_1} + \pi_{2,t} {X}_{t}^{x_2, u_2} $, 
$\kappa_t {X}_t^{x, u} = \kappa_{1,t} {X}_{t}^{x_1, u_1} + \kappa_{2,t} {X}_{t}^{x_2, u_2} $, and $D_t = D_{1,t} + D_{1,t}$, which together imply $X_t^{x,u} = X_{t}^{x_1, u_1} + X_{t}^{x_2, u_2}$.
By the definition of $\oplus$ and the linearity of \eqref{eq:dX}, we have $u_1 \oplus u_2 \in \Ac(x_1 + x_2)$ and $V(x_1) + V(x_2) \le V(x_1 + x_2)$. 

For any $u \in \Ac(x)$ and $\epsilon >0$, define a \emph{perturbed} objective functional $J_\epsilon$ by 
\begin{align}
\label{eq:J_eps}
J_\epsilon(x; u) :=  \Eb \left[\int_0^\infty \, e^{-\delta t} \, \frac{\big(D_t + \epsilon \dione_t \big)^{1-\eta}}{1 - \eta} \, \dd t \right] = J(x+\epsilon; u \oplus u_c^*),
\end{align}
where the second equality comes from the definitions of the operator $\oplus$ above and $J$ in \eqref{eq:J}. 
Let us define the corresponding value function $V_\epsilon$ by 
\begin{align}
\label{eq:V_eps}
V_\epsilon(x) := \sup_{u \in \Ac(x)} \, J_\epsilon(x; u),
\end{align}
where $J_\epsilon$ is defined in \eqref{eq:J_eps}.

\subsection{Main Results}
\label{sub:main}

We now present the main result of this paper in Theorem \ref{thm:main}. 

\begin{theorem}
\label{thm:main}
Suppose $U(x) = x^{1 - \eta}/(1-\eta)$, where $\eta > 0$ and $ \eta \neq 1$, $C$ defined in \eqref{eq:ABC} is non-negative, and $\psi$ defined in \eqref{eq:psi} is positive. 
For any $\epsilon>0$ and $x>0$, we have 
\begin{align}
\label{eq:main}
V_\epsilon(x) = V_c(x+\epsilon),
\end{align}
where $V_\epsilon$ is defined in \eqref{eq:V_eps} and $V_c$ is obtained in \eqref{eq:V_pow}.
\end{theorem}

\begin{proof}
Using the definition of $J_\epsilon$ in \eqref{eq:J_eps}, we easily see that $J_\epsilon(x; u_c^*) = J(x+\epsilon; u_c^*) = V_c(x+\epsilon)$, where $u_c^*$ is the optimal constant strategy derived in \eqref{eq:pow_cont_op}. 
Since $u_c^* \in \Ac(x)$, we obtain $V_\epsilon(x) \ge V_c(x + \epsilon)$. 
In the remaining of the proof, we aim to show the converse inequality, $V_\epsilon(x) \le V_c(x+\epsilon)$, also holds.

Let $\epsilon>0$ and $x>0$ be given, take any admissible control $u \in \Ac(x)$. 
We define a new control $u^\epsilon := u \oplus u_c^*$,  where the initial wealth associated with the strategy $u_c^*$ is $\epsilon$, and denote it by $u^\epsilon = (\pi^\epsilon, \kappa^\epsilon, D^\epsilon)$.
Introduce the corresponding wealth process by $X^\epsilon = (X^\epsilon_t)_{t \ge 0}$, i.e., $X_t^{\epsilon} := X_t^{x+\epsilon, u^\epsilon}$. 
%Denote by $D^\epsilon$ the dividend strategy of $u^\epsilon$.
Recall from \eqref{eq:X_cont} that the wealth process under a constant strategy is always positive, which leads to 
\begin{align}
\label{eq:positive}
X_t^\epsilon = X_t^{x,u} + \epsilon \xone_t \ge \epsilon \xone_t > 0, \qquad \forall \, t \ge 0.
\end{align}

Define another new process $M^\epsilon$ by 
\begin{align}
\label{eq:M}
M_t^\epsilon := \int_0^t \, e^{-\delta s} U(D^\epsilon_s) \, \dd s + e^{-\delta t} \, V_c(X_t^\epsilon), \qquad \forall \, t \ge 0.
\end{align}
Recall from  \eqref{eq:V_pow} that $V_c(x) = \psi^{-\eta} x^{1-\eta} / (1-\eta)$.
Since $V_c$ is a smooth function over $(0, \infty)$ and $X^\epsilon >0$, we can apply It\^o's lemma to $M^\epsilon$.
Using 
%$V_c'(x) = (\psi x)^{-\eta}$ and $V_c''(x) = - \eta \psi^{-\eta} x^{-1-\eta}$, and 
the SDE \eqref{eq:dX} of $X^\epsilon$, we obtain 
\begin{align}
	\label{eq:dM} \qquad
\dd M_t^\epsilon \, = \, & e^{-\delta t} V_c(X_{t-}^\epsilon) \left[(\sig \pi_t^\epsilon - \beta \rho \kappa_t^\epsilon) \, \dd \wone_t - \beta \sqrt{1 - \rho^2} \kappa_t^\epsilon \, \dd \wtwo_t + \big( (1-\gam \kappa_t^\epsilon)^{1-\eta} -1\big) \, \dd \wt{N}_t\right] + \\
& e^{-\delta t} \Bigg\{ - \delta V_c(X_{t}^\epsilon) + U(D_t^\epsilon) + V_c'(X_{t}^\epsilon) \big[ (r + (\mu - r)\pi_t^\epsilon + (p - \alpha) \kappa_t^\epsilon ) X_t^\epsilon - D_t^\epsilon \big] \\
&\qquad  \frac{1}{2} V_c''(X_{t}^\epsilon) \big(\sig^2 (\pi_t^\epsilon)^2 - 2 \beta \rho \sig \pi_t^\epsilon \kappa_t^\epsilon + \beta^2 (\kappa_t^\epsilon)^2 \big) (X_t^\epsilon)^2 + \lam \big( V_c((1 - \gam \kappa_t^\epsilon) X_{t-}^\epsilon) - V_c( X_{t-}^\epsilon)\big)
\Bigg\} \, \dd t.
\end{align}

We decompose the $\dd t$ term in \eqref{eq:dM} into two parts as $e^{-\del t} (Y_t^{(1)} + Y_t^{(2)}) \dd t$, where 
\begin{align}
Y_t^{(1)} &:=   U(D_t^\epsilon) - D_t^\epsilon \, V_c'(X_{t}^\epsilon) - \frac{\eta}{1 - \eta} \, \left(V_c'(X_{t}^\epsilon)\right)^{1- \frac{1}{\eta}}, 
%Y_t^{(2)} &:=  \frac{\eta}{1 - \eta} \, \left(V_c'(X_{t}^\epsilon)\right)^{1- \frac{1}{\eta}} - \delta V_c(X_{t}^\epsilon) + X_t^\epsilon  \, V_c'(X_{t}^\epsilon) \big(r + (\mu - r)\pi_t + (p - \alpha) \kappa_t \big) 
\end{align}
and $Y_t^{(2)}$ is the remaining part of the $\dd t$ term in \eqref{eq:dM}.
Using the first-order condition, we show that the inequality $\frac{1}{1 - \eta} \, y^{1-\eta} - y - \frac{\eta}{1 - \eta} \le 0$ holds true for all $y \ge 0$ (the maximum value is 0 taken at $y=1$).
%\begin{align}
%\frac{1}{1 - \eta} \, y^{1-\eta} - y - \frac{\eta}{1 - \eta} \le 0
%\end{align}
By substituting $D_t^\epsilon \, (V_c'(X_{t}^\epsilon))^{1/\eta}$ for $y$ in the above inequality, we obtain that $Y_t^{(1)} \le 0$ for all $t \ge 0$. 
Recall the result of $V_c$ in \eqref{eq:V_pow}, from which we get 
$V_c'(x) = (\psi x)^{-\eta}$ and $V_c''(x) = - \eta \psi^{-\eta} x^{-1-\eta}$.
We then analyze $Y_t^{(2)}$ as follows
\begin{align}
Y_t^{(2)} &:= \frac{\eta}{1 - \eta} \, \left(V_c'(X_{t}^\epsilon)\right)^{1- \frac{1}{\eta}} - \delta V_c(X_{t}^\epsilon) + X_t^\epsilon  \, V_c'(X_{t}^\epsilon) \big(r + (\mu - r)\pi_t^\epsilon + (p - \alpha) \kappa_t^\epsilon \big) \\
&\qquad \frac{1}{2} V_c''(X_{t}^\epsilon) \big(\sig^2 (\pi_t^\epsilon)^2 - 2 \beta \rho \sig \pi_t^\epsilon \kappa_t^\epsilon + \beta^2 (\kappa_t^\epsilon)^2 \big) (X_t^\epsilon)^2 + \lam \big( V_c((1 - \gam \kappa_t^\epsilon) X_{t-}^\epsilon) - V_c( X_{t-}^\epsilon)\big) \\
&\;= (g(\pi_t^\epsilon, \kappa_t^\epsilon) - \mathfrak{g}^*) \cdot (1-\eta) V_c(X_{t}^\epsilon) = (g(\pi_t^\epsilon, \kappa_t^\epsilon) - \mathfrak{g}^*) \cdot \psi^{-\eta} (X_t^\epsilon)^{1-\eta},
\end{align} 
where we have used the definitions of $g$ in \eqref{eq:g} and $\psi$ in \eqref{eq:psi} to derive the second equality. Since $\mathfrak{g}^*$ is the maximum value of the function $g$ over $\Rb \times [0, 1/\gam)$, we obtain $Y_t^{(2)} \le 0$ for all $t \ge 0$. 
In addition, we notice that  $Y_t^{(1)} = Y_t^{(2)} = 0$ if and only if $u^\epsilon = u_c^*$.

Let us denote $L^\epsilon$ the local martingale part of $ M_t^\epsilon$ defined in \eqref{eq:M}.
%By Definition \ref{def:ad}, $\pi$ is square integrable and $\kappa \in [0, 1/\gam)$. 
%By imposing a stronger constraint $\kappa \in [0, 1/\gam - \epsilon']$, where $\epsilon'>0$ is small enough.
We claim that  $L^\epsilon$ is a supermartingale. 
To see this result, we  define $M_0^\epsilon$ similar to that of $M^\epsilon$ in \eqref{eq:M}, but under $x=0$ and $u \equiv 0$.
By repeating the above analysis, we easily show that $M_0^\epsilon$ is a uniformly integrable martingale (recall with $u \equiv 0$, $u^\epsilon = u_c^*$).
Now using the monotonicity of both $U$ and $V_c$, we deduce that $M^\epsilon \ge M_0^\epsilon$, which immediately proves that $L^\epsilon$ is indeed a supermartingale as claimed.
By the previous finding on $Y_t^{(1)}$ and $Y_t^{(2)}$, we get
\begin{align}
\label{eq:in-M}
\Eb \left[M_t^\epsilon \right] \le V_c (x + \epsilon) + \Eb \left[ L_t^\epsilon \right] \le V_c(x + \epsilon).
\end{align}

In the final step, using \eqref{eq:J_eps}, \eqref{eq:positive}, \eqref{eq:M}, and the above results, we obtain 
\begin{align}
J_\epsilon(x; u) &= \lim_{t \to \infty} \, \Eb \left[\int_0^t \, e^{-\del s} \, U(D_s^\epsilon) \, \dd s\right] = \lim_{t \to \infty} \,  \Eb \left[ M_t^\epsilon - e^{-\del t} \, V_c(X_t^\epsilon)\right] \\
&\le \limsup_{t\to \infty} \Eb\left[ M_t^\epsilon \right] - \liminf_{t \to \infty} \Eb \left[ e^{-\del t} \, V_c(X_t^\epsilon)\right] \\
&\le \limsup_{t\to \infty} \Eb\left[ M_t^\epsilon \right] \le V_c(x + \epsilon),
\end{align}
where, to prove the second inequality above, we have used the following result
\begin{align}
	\label{eq:non-negative}
\liminf_{t \to \infty} \Eb \left[ e^{-\del t} \, V_c(X_t^\epsilon)\right] 
%= \psi^{-\eta} \, \liminf_{t \to \infty} \Eb \left[ e^{-\del t} \, \frac{(X_t^u + \epsilon X_t^{(1)})^{1-\eta}}{1 - \eta}\right] 
\ge \psi^{-\eta} \epsilon^{1-\eta} \, \liminf_{t \to \infty} \Eb \left[ e^{-\del t} \, \frac{( X_t^{(1)})^{1-\eta}}{1 - \eta}\right] = 0.
\end{align}
By taking supremum over $\Ac(x)$, we obtain $V_\epsilon(x) \le V_c(x + \epsilon)$. 
The proof is now complete.
\end{proof}

%\begin{remark}
%The proof to Theorem \ref{thm:main} is modified from \cite{herdegen2020elementary} (see Theorems 4.3 and 6.1 therein), with noticeable simplifications on handling the $\dd t$ term in \eqref{eq:dM}.
%The arguments on $L^\epsilon$ being a supermartingale can also be simplified by noting that $\kappa$ is bounded here (see Definition \ref{def:ad}).
%\end{remark}

Using Theorem \ref{thm:main} and the monotonicity result $J(x;u) \le J_\epsilon(x; u)$, we have the following corollary that eventually verifies the conjecture in the opening of this section.

\begin{corollary}
\label{cor:veri}
Under the same assumptions as in Theorem \ref{thm:main}, we have 
\begin{align}
	V(x) = V_c(x) = J(x, u_c^*), \qquad \forall \, x>0,
\end{align}
where $u_c^*$ and  $V_c$  are given respectively by \eqref{eq:pow_cont_op} and \eqref{eq:V_pow}, and $V$ is the value function to Problem \eqref{eq:prob}.
\end{corollary}

In both Theorem \ref{thm:main} and Corollary \ref{cor:veri}, the utility function is of power form, with $\eta \neq 1$. 
When $\eta = 1$ (the log utility case), the same result holds, i.e., we still have $V(x) = V_c(x)$, where now $V_c$ is given by \eqref{eq:log_cont_value} from Proposition \ref{prop:log}. 
The analysis leading to this conclusion is similar to the one given above for the power utility case, and is thus omitted.

\begin{corollary}
\label{cor:veri_log}
Suppose $U(x) = \ln x$ and the constant $C$ defined in \eqref{eq:ABC} is non-negative.  We have 
\begin{align}
	V(x) = V_c(x) = J(x, u_c^*), \qquad \forall \, x > 0,
\end{align}
where $u_c^*$ and  $V_c$  are given respectively by \eqref{eq:log_cont} and \eqref{eq:log_cont_value}, and $V$ is the value function to Problem \eqref{eq:prob}.
\end{corollary}

\subsection{Comparisons with the HJB Approach}
\label{sub:comp}

In this subsection, we compare the perturbation approach adopted in this paper (see \cite{herdegen2020elementary}) with the standard HJB approach of solving stochastic control problems (see, e.g., \cite{merton1969lifetime, merton1971optimum}). 

The perturbation approach  works brilliantly for the insurer's optimal dividend problem under  HARA (power and log) utility considered in Section \ref{sub:main}. In addition, the proof to Theorem \ref{thm:main} applies to all the cases of the relative risk aversion $\eta$: $0 < \eta <1$, $\eta = 1$, and $\eta > 1$. 
In comparison,  the HJB approach encounters technical issues when $\eta > 1$.
For this reason, many related existing works that rely on the HJB approach assume $0 < \eta \le 1$; see \cite{cadenillas2007optimal}, \cite{thonhauser2011optimal}, \cite{jin2015optimal}, and \cite{xu2020optimal}. 
Below we discuss some of these issues in details.

\begin{itemize}
	\item A key step in the verification theorem under the HJB approach is to guarantee the validness of interchanging the order of a limit and a conditional expectation (integral) for a family of processes when passing the limit to infinity. 
	A sufficient condition for such an interchange is the uniform integrability of the processes, which can be  achieved by imposing the growth condition on the value function or the utility function. 
	For instance, 
	%\cite{korn2001option}[Theorem 5.17, p.229] assume that $|V(t,x)| \le K(1 + |x|^k)$ holds for the candidate value function $V(t,x)$, where $K>0$ and  $ k \in \Nb$; 
	\cite{capponi2014dynamic}[Eq.(4.10)]
	%\cite{sotomayor2009explicit}[Eq.(2.2)] 
	assume that $U(x) \le K(1 + x)$ holds for the utility function $U$, where $K,x>0$.
	%; see also a similar condition in .
	Alternatively, one may ``simply" impose a much stronger condition of uniform boundedness on both the drift and the diffusion terms;  see \cite{fleming2006controlled}[Section IV.5, Eq.(5.2), p.164]. 
	It is trivial to see that both ``solutions" are restrictive and do not apply to Problem \eqref{eq:prob} when $\eta > 1$. 
	
	\item The second technical issue we discuss is the possibility of bankruptcy (i.e., $X_t^{x,u} = 0$ at some time $t < \infty$). A full and rigorous treatment of bankruptcy under the standard approaches requires lengthy technical arguments. \cite{karatzas1986explicit} introduce a free parameter $P$ as the bequest value at the bankruptcy time and consider a modified version of the original problem (which has a boundary condition $\lim_{x \downarrow 0} V^P(x) = P$), and take pages to show that the value function $V^P(x)$ of the modified problem converges to the value function $V(x)$ of the original problem as $P \to -\infty$.
	
	\item Last, a commonly required condition for an infinite horizon problem is the so-called transversality condition, which is given by 
	 \begin{align}
	 	\label{eq:trans}
		\lim\limits_{t \to \infty} \, \Eb \left[e^{-\delta t} \, \frac{(X_t^u)^{1-\eta}}{1 - \eta }\right] = 0 \quad \text{or in a weaker version} \quad 
		\liminf\limits_{t \to \infty} \, \Eb \left[e^{-\delta t} \, \frac{(X_t^u)^{1-\eta}}{1 - \eta }\right] \ge  0.
	\end{align}
	The above transversality condition \eqref{eq:trans} is satisfied automatically if $0 < \eta <1$ but may fail in general when $\eta > 1$. Proving \eqref{eq:trans} is certainly not trivial when $\eta > 1$, although many applications ``simply" assume \eqref{eq:trans} holds to avoid a proof.
	In fact, \eqref{eq:trans} may not be expected; see \cite{herdegen2020elementary}[Remark 4.7].
\end{itemize}
Further technical discussions can be found in \cite{herdegen2020elementary}. 
We remark that another standard approach, the martingale (duality) method, faces the ``dual'' side of the technical issues/assumptions in the HJB approach. 
For that, we refer readers to  \cite{karatzas1998methods}[Section 3.9] for technical assumptions imposed on optimal investment problems over an infinite horizon.

Having seen some of the technical difficulties 
%linked with the case of $\eta > 1$,
under the HJB approach, we now explain why the perturbation approach proposed in \cite{herdegen2020elementary} works so well in the proof of Theorem \ref{thm:main}. First, by \eqref{eq:positive}, the perturbed wealth process $X^\epsilon$ is \emph{strictly} positive, immediately yielding a significant advantage. 
That is, we no longer need to deal with bankruptcy in the proof of Theorem \ref{thm:main}.
Second, by utilizing the explicit results from Section \ref{sec:cont} and the monotonicity of the utility/value functions, we easily show that $M_0^\epsilon$ is uniformly integrable and $M^\epsilon_t \ge M^{\epsilon}_{0,t}$ for all $t \ge 0$. 
(Note $M^\epsilon$ is defined by \eqref{eq:M} and $M^\epsilon_0$ is a special version of $M^\epsilon$ under $x=0$ and $u=0$.)
There results together verify that the process $L^\epsilon$, the local martingale part of $M^\epsilon$, is a supermartingale, which leads to an essential inequality in \eqref{eq:in-M}. 
In other words, a carefully-chosen special (optimal) strategy enables us to obtain a lower bound on the process $M^\epsilon$, bypassing the difficulty of proving uniform integrability of $(e^{-\delta t} V(X_t^u))_{t \ge 0}$. 
Note that in the  definition of $M^\epsilon$ in \eqref{eq:M}, the second term is $e^{-\delta t} V_c(X_t^\epsilon)$, not $e^{-\delta t} V(X_t^u)$ as in the standard HJB approach. 
Here the advantages are at least twofold: (1) $V_c$ is already obtained explicitly in \eqref{eq:V_pow} and is smooth, while the value function $V$ is assumed to be smooth and is yet to be solved from  the associated HJB equation, and (2) $X_t^\epsilon >0$ but $X_t^u \ge 0$ for all $t$.
Last, we point out that we derive an inequality in \eqref{eq:non-negative}, which is in the same spirit to the  transversality condition in \eqref{eq:trans}. 
To be precise,  replacing $X^u$ by $X^\epsilon$ in \eqref{eq:trans} leads to \eqref{eq:non-negative}.
In our case, proving \eqref{eq:non-negative} is very easy, by the monotone increasing property of $V_c$ and \eqref{eq:positive}.   
However, as already pointed out previously, the transversality condition \eqref{eq:trans} does not hold in general, and even for specific problems when it does hold, proving it is not trivial in the case of $\eta > 1$ due to the possibility of bankruptcy.

In the last part, we discuss the shortcomings of the perturbation approach. 
The standard HJB and martingale approaches are rather general, and can be applied to solve a wide range of control problems even beyond utility maximization. 
However, the perturbation approach used in \cite{herdegen2020elementary} and this paper is on the special side; these two papers show that this approach provides an extraordinary alternative to optimal investment/consumption and optimal dividend problems under the HARA utility maximization. 
It remains an open question to explore the use of this approach to different control problems.

\section{Economic Analysis}
\label{sec:econ}

In this section, we conduct an economic analysis to study how the model parameters and risk aversion affect the insurer's optimal strategy $u^*$. 
Due to the popularity of applying a diffusion process to approximate the risk model (see \cite{browne1995optimal} and \cite{hojgaard1998optimal}), we assume there are no jumps in the risk process $R$  (setting $\lam =0$ in \eqref{eq:dR}), unless stated otherwise. 
Recall that, in the case of no jumps,  we obtain the optimal strategy $u^* = u_c^*$, given by \eqref{eq:cor_op}, in Corollary \ref{cor:power}. 
We divide the economic analysis along two directions, with analytic results in Section \ref{sub:analytic} and numerical results in Section \ref{sub:nume}.

\subsection{Analytic Results}
\label{sub:analytic}

In the insurance literature, one often assumes that the insurance market is \emph{independent} of the financial market, i.e., $\rho = 0$ in \eqref{eq:dR}. 
But as argued in \cite{stein2012stochastic} and many others, ignoring the possible dependence between the two markets could lead to catastrophic consequences (e.g., the infamous AIG case in the financial crisis of 2007-2008). 
Hence in the first study, we investigate the impact of $\rho$ on the insurer's optimal strategy.
By \eqref{eq:cor_op}, we obtain 
\begin{align}
\label{eq:rho}
\frac{\partial \pi_c^*}{\partial \rho} =  \frac{\beta}{\sig}\, \kappa_c^*>0, 
\qquad 
\frac{\partial \kappa_c^*}{\partial \rho} = \frac{\Lam}{\eta  \beta} + 2 \rho \kappa_c^*,
\qquad 
\frac{\partial \xi_c^*}{\partial \rho} = - \frac{1-\eta}{\eta^2} \, \left[ \beta \Lam \kappa_c^* + \rho \big(\kappa_c^*\big)^2 \right],
\end{align}
where the Sharpe ratio $\Lam = (\mu - r) / \sig$ is defined in \eqref{eq:Lam}. 
An interesting result is that $\partial \pi_c^* / \partial \rho > 0$, 
which is also found in \cite{zou2014optimal}[Figures 1-3] and \cite{shen2020mean}[Figure 1].
%Using this result, we may interpret $\rho$ as an important investment signal: the higher the $\rho$, the better the investment opportunity of the risky asset. 
%
By \eqref{eq:cor_op} and \eqref{eq:rho}, when $\rho >0$, we observe $\partial \kappa_c^* / \partial \rho>0$ and $\text{Sign} \, (\partial \xi_c^* / \partial \rho)  = \text{Sign} \, (\eta - 1)$, and further $\pi_c^* >0$.
Note that when $\rho >0$ the risky asset provides a natural hedge to the insurance business: an increase in the risk $R$ (losses to the insurer) is associated with an increase in the risky asset price $S$. This observation implies that the insurer should long the risky asset, explaining why $\pi_c^* |_{\rho >0} >0$. 
For the same reason, when the \emph{positive} correlation becomes stronger, the hedging effect amplifies, or equivalently, the ``risky'' insurance business becomes ``less risky", indicating $\frac{\partial \kappa_c^*}{\partial \rho}|_{\rho >0}>0$. 
Regarding the result of $\partial \xi_c^* / \partial \rho$ in \eqref{eq:pi}, observe that the inequality $\beta \Lam \kappa_c^* + \rho (\kappa_c^*)^2 >0$ holds if and only if $\rho > - \beta \Lam / \kappa_c^*$, which can  be seen as a good indication of the overall market condition.\footnote{%Roughly speaking, $\beta \Lam \kappa_c^* + \rho (\kappa_c^*)^2 >0$ holds when $\rho$ is not too negative. 
Recall that $\pi_c^* = \frac{\mu - r}{\eta \sig^2} + \frac{\rho \beta }{\sig} \kappa_c^*$, where all the components are positive except the possibility of $\rho$ being negative. When $\rho$ is negative, the second term of hedging demand becomes negative and acts on the opposite direction of the first term  (Merton's strategy) in $\pi_c^*$. 
When $\rho \le - \beta \Lam / \kappa_c^*$, the negative second term  dominates the positive first term in $\pi_c^*$, i.e., the insurer needs to ``short sell'' the risky asset mainly for the hedging purpose, which may lead to excessive risk taking.}
Next, recall that the power utility function in \eqref{eq:U} is applied to quantitatively measure the welfare of dividend payments (monetary wealth). In similar settings, empirical evidence often shows that the relative risk aversion $\eta$ is greater than 1; see for instance \cite{meyer2005relative}[Table 1]. 
As a result, it is reasonable for us to assume $\eta > 1$, although we do not rule out the possibility of $0 < \eta \le 1$.  
Now when $\rho$ increases in the region $(- \beta \Lam / \kappa_c^*, 1)$, the overall market condition is likely to become more favorable to the insurer. As such, insurers with high risk aversion ($\eta > 1$) pay dividend at a higher rate, since they prefer to ``cash out'' now while the market is still in ``good'' regime;  on the other hand, insurers with low risk aversion ($0<\eta <1$) reduce the dividend rate, as the desire to invest more capital in the risky asset is dominating. 
%To summarize, under an ideal market condition of $\rho >0$ and a reasonable risk aversion $\eta >1$, all the optimal investment, insurance, and dividend strategies have a positive relation with the correlation parameter $\rho$.
We end this study by commenting that our theoretical result in \eqref{eq:rho} offers insight on $\eta > 1$ as shown in the empirical findings of \cite{meyer2005relative}. 
To see this, let us assume for a moment that $\rho >0$. Then as $\rho$ increases, the overall market improves for  the insurer and naturally we expect all the insurer's strategies $\pi_c^*$, $\kappa_c^*$, and $\xi_c^*$ to increase at the same time; while \eqref{eq:rho} shows that $\xi_c^*$ increases only if $\eta > 1$. 

As seen in the above analysis on the correlation coefficient $\rho$, the insurer's (relative) risk aversion parameter $\eta$  is decisive in the comparative statics of the optimal dividend strategy. 
We next analyze how risk aversion $\eta$ affects the insurer's optimal strategy. 
To that end, we obtain from \eqref{eq:cor_op} that
\begin{align}
	\label{eq:eta}
	\frac{\partial \pi_c^*}{\partial \eta} =  - \frac{1}{\eta}\, \pi_c^*, 
	\qquad 
	\frac{\partial \kappa_c^*}{\partial \eta} = - \frac{1}{\eta}\, \kappa_c^* <0,
	\qquad 
	\frac{\partial \xi_c^*}{\partial \eta} = - \frac{1}{\eta} \, \xi_c^* + \frac{\hat{\mathfrak{g}}^* + (\eta - 1)r}{\eta^2}.
\end{align}
Due to \eqref{eq:cor_cond}, $\kappa_c^*>0$ and that implies a higher risk aversion always leads to a reduction of underwriting in the insurance business. 
Since $\text{Sign} (\partial \pi_c^* / \partial \eta) = - \, \text{Sign}(\pi_c^*)$, when the optimal investment strategy is a ``buy'' strategy (resp. a ``short-sell'' strategy), the investment proportion in the risky asset reduces (resp. increases) as the risk aversion $\eta$ increases. 
Notice that when $\pi_c^*<0$, an increase in $\pi_c^*$ means the absolute value of $\pi_c^*$ decreases. 
Therefore, given an increase in the risk aversion $\eta$, the insurer always invests less proportion in absolute values in the risky asset (i.e., less risk taking).
Such a result is clearly consistent with the economic meaning of risk aversion.
Despite the analytic result of $\partial \xi_c^* / \partial \eta$ in \eqref{eq:eta}, how risk aversion affects the dividend strategy is still unclear, since the first term is always negative but the second term is always positive (note $\hat{\mathfrak{g}}^* - r >0$ by \eqref{eq:g-hat}).

In the third analysis, we focus on the impact of the financial market on the optimal strategy. 
Introduce the excess return $\bar \mu$ and the excess premium $\bar p$ by
\begin{align}
	\label{eq:mu-bar}
\bar \mu := \mu - r \qquad \text{ and } \qquad \bar p := p - \alpha,
\end{align}
where $\bar \mu>0$ and $\bar p>0$ by \eqref{eq:assu}.
To see how the excess return $\bar \mu$ affects the insurer's decision, we compute 
\begin{align}
\label{eq:pi}
	\frac{\partial \pi_c^*}{\partial \bar \mu} =  \frac{1}{\eta \sig^2  ( 1 - \rho^2)} > 0, 
	\qquad 
	\frac{\partial \kappa_c^*}{\partial \bar \mu} = \frac{\rho}{\eta \beta \sig  (1 - \rho^2)},
	\qquad 
	\frac{\partial \xi_c^*}{\partial \bar \mu} = - \frac{1-\eta}{\eta^2} \, \frac{\rho   \bar p + \beta \Lam}{\beta \sig (1 - \rho^2)}.
\end{align}
By \eqref{eq:pi}, following an increase in $\bar \mu$, the insurer always invests more in the risky asset but only underwrites more (resp. less) insurance policies if $\rho > 0$ (resp. $\rho <0$). 
These findings can be explained using the discussions of $\rho$  in the above study and the ``fact'' that the higher the excess return $\bar \mu$, the more attractive the risky asset. 
Note that in our framework the insurer is exposed to both the investment risk from the risky asset $S$ and the insurable risk $R$. But according to \eqref{eq:pi}, $\partial \pi_c^* / \partial \bar \mu >0$ always holds, regardless of the sign of $\rho$. 
That is, even though we allow a non-zero correlation $\rho$ in our model, the result $\partial \pi_c^* / \partial \bar \mu >0$ is still in line with the standard optimal investment literature; see \cite{merton1969lifetime, merton1971optimum} and \cite{herdegen2020elementary}.
We mention that $\rho$ impacts the sensitivity magnitude of $\pi_c^*$ with respect to $\bar \mu$: the stronger the correlation, the more sensitive $\pi_c^*$ is to the changes of  $\bar \mu$.
Assuming $\rho  \bar p + \beta \Lam >0$,\footnote{A sufficient condition for this inequality is $\rho \ge 0$, as $\bar p, \beta, \Lam >0$.}  as $\bar \mu$ increases, insurers with low risk aversion ($0<\eta <1$) reduce the dividend rate but  insurers with high risk aversion ($\eta > 1$) reacts exactly in the opposite way.
But, if $\rho  \bar p + \beta \Lam < 0$,\footnote{Equivalently, $\rho<0$ and $\bar p > - \beta \Lam /\rho >0$. The economic meaning of this condition is that, in an ``adverse" market (since correlation is negative), the insurer sets the net insurance premium above a threshold.} the converse of the above statement holds true.
Another key parameter in the financial market is the volatility $\sig$ of the risky asset, and the related sensitivity results are obtained by 
\begin{align}
	\label{eq:sig}
	\frac{\partial \pi_c^*}{\partial \sig} = -  \frac{(2 - \rho^2) \Lam  }{\eta \sig^2 (1 - \rho^2)} - \frac{\rho \beta}{\sig^2} \, \kappa_c^*, 
	\qquad 
	\frac{\partial \kappa_c^*}{\partial \sig} = - \frac{\rho \bar \mu}{\eta \beta \sig^2  (1 - \rho^2)},
	\qquad 
	\frac{\partial \xi_c^*}{\partial \sig} = \frac{1-\eta}{\eta^2} \, \frac{\bar \mu}{\sig^2}  \, \frac{\rho \bar p + \beta \Lam}{\beta (1 - \rho^2)}.
\end{align}
If we set $\rho =0$ in \eqref{eq:sig}, we have $\partial \pi_c^* / \partial \sig = - 2\Lam /(\eta \sig^2) <0$, but such a result is \emph{not} expected in general, which is different from the standard literature; see \cite{merton1969lifetime, merton1971optimum}.
From \eqref{eq:sig}, we observe that $\partial \pi_c^* / \partial \sig <0 $ whenever $\rho \ge 0$ 
and  $\text{Sign}(\partial \kappa_c^* / \partial \sig) = - \, \text{Sign}(\rho)$.
Consequently, assuming the two markets are positively correlated, when the volatility $\sig$ increases, the insurer invests less in the risky asset and reduces insurance liabilities. 
Under a negatively correlated condition, the insurer's reaction to an increase of the volatility $\sig$ is to underwrite more policies.
However, given $\rho<0$, how the insurer should adjust her investment strategy to an increase of $\sig$ is not easily seen, and will be investigated numerically in the next subsection. 
Comparing \eqref{eq:mu-bar} with \eqref{eq:sig} shows that $\partial \xi_c^* / \partial \sig = - \Lam \, \partial \xi_c^* / \partial \bar \mu$, and thus the opposite side of the previous discussions on $\partial \xi_c^* / \partial \bar \mu$ applies here.

Our last agenda is to study the impact of the insurance market, the drift parameter $\alpha$ and the diffusion parameter $\beta$ in the risk process \eqref{eq:dR}, on the insurer's optimal strategy.
We first focus on the optimal liability strategy $\kappa_c^*$ and obtain 
\begin{align}
\label{eq:sen-kappa}
\frac{\partial \kappa_c^*}{\partial \alpha} = - \frac{1}{\eta \beta^2 (1 - \rho^2)} <0 \qquad \text{ and } \qquad \frac{\partial \kappa_c^*}{\partial \beta} = - \frac{2 \bar p}{\eta \beta^3 (1 - \rho^2)}<0.
\end{align}
The sensitivity results in \eqref{eq:sen-kappa} fit our intuition perfectly: when $\alpha$ or $\beta$ increases, the insurable risk $R$ becomes more risky  and the insurer should reduce underwriting in response. Regarding the optimal investment strategy $\pi_c^*$, we have 
\begin{align}
	\label{eq:sen-pi}
	\frac{\partial \pi_c^*}{\partial \alpha} = \frac{\rho \beta }{\sig} \, \frac{\partial \kappa_c^*}{\partial \alpha} \qquad \text{ and } \qquad \frac{\partial \pi_c^*}{\partial \beta} =\frac{\rho \beta }{\sig} \, \left(\frac{\kappa_c^*}{\beta} + \frac{\partial \kappa_c^*}{\partial \beta} \right) .
\end{align}
An immediate result is that $\text{Sign}(\partial \pi_c^* / \partial \alpha) = \text{Sign}(\rho)$, which is consistent with the preceding analysis on $\rho$. 
But the story of $\partial \pi_c^* / \partial \beta$ in \eqref{eq:sen-pi} is more complex; it not only depends on the correlation $\rho$ but also $\kappa_c^*$ (positive by \eqref{eq:cor_cond}) and its derivative with respect to $\beta$ (negative by \eqref{eq:sen-kappa}). 
Last, we analyze the optimal dividend strategy $\xi_c^*$ and obtain 
\begin{align}
	\label{eq:sen-xi}
\frac{\partial \xi_c^*}{\partial \alpha} = \frac{1 - \eta}{\eta} \, \kappa_c^* \qquad \text{ and } \qquad \frac{\partial \xi_c^*}{\partial \beta} = \frac{1 - \eta}{\eta^2} \,  \frac{\bar{p}^2 + \beta \rho \bar p \Lam}{\beta^3 (1 - \rho^2)} .
\end{align}
From \eqref{eq:sen-xi}, we easily see that $\text{Sign}(\partial \xi_c^* / \partial \alpha) = \text{Sign}(1-\eta)$ and  $\text{Sign}(\partial \xi_c^* / \partial \beta) = \text{Sign}(1-\eta)$ if $\bar p > \max\{0, -\beta \rho \Lam\}$. 
Now if we assume $\eta > 1$ and $p > \max\{0, -\beta \rho \Lam\}$, both derivatives in \eqref{eq:sen-xi} are negative, implying that the insurer should reduce the optimal dividend payment rate when the insurance business becomes more risky.
This finding may further support that $\eta > 1$ is a more reasonable scenario.

\subsection{Numerical Results} 
\label{sub:nume}

In Section \ref{sub:analytic}, we analyze the sensitivity of the optimal strategy with respect to various model parameters \emph{analytically}. 
Here, in this subsection, we continue the same analysis, but from a numerical point of view.
In particular, we focus on the effect of two important parameters: the correlation coefficient $\rho$ and the relative risk aversion $\eta$.

\begin{table}[h!]
	\centering
	\caption{Default Model Parameters}
	\label{tab:para}
	\vspace{1ex}
	\begin{tabular}{cccccccc} \hline
		Parameters & $r$ & $\mu$ & $\sig$ & $\alpha$ & $\beta$ & $p$  & $\delta$\\ \hline
		Values  & 0.01 & 0.05 & 0.25 & 0.1 & 0.1 & 0.15  & 0.15\\ \hline
	\end{tabular}
	
	\vspace{1ex}
	{\footnotesize Note. $\delta$ is the subjective discount factor, $p$ is the premium rate, and all other parameters are from the model \eqref{eq:dS} and \eqref{eq:dR}.}
\end{table}

Since our sensitivity analysis is \emph{qualitative}, we set the base values for the model parameters of \eqref{eq:dS}-\eqref{eq:dR} in Table \ref{tab:para} and allow \emph{one} parameter to vary over a reasonable range in each study.
We compute the insurer's optimal investment, liability ratio, and dividend rate strategies $u^*=(\pi^*, \kappa^*, \xi^*)$, where $u^* = u^*_c$ in \eqref{eq:cor_op}, when the correlation coefficient $\rho$ varies from -0.8 to 0.8. 
We plot the results under four different risk aversion levels ($\eta = 0.8,1,2,5$) in Figure \ref{fig:rho}.
Our main findings are summarized as follows:

\begin{itemize}
	\item (Optimal investment $\pi^*$ in the upper panel of Figure \ref{fig:rho}) 
	
	We observe that $\pi^*$ is an increasing function of $\rho$, which verifies the result in \eqref{eq:rho}.  
	In consequence, the insurer invests more in the risky asset as $\rho$ increases.
	There exists a \emph{negative} threshold $\hat \rho$\footnote{We use $\hat \rho$ to denote a genetic threshold constant of $\rho$, which may differ content by content.} ($\hat \rho$ is about -0.32 in the numerical example), below which the optimal investment $\pi^*$ is negative, i.e., short selling is optimal when $\rho < \hat \rho$. 
	We also notice that the absolute investment weight $|\pi^*|$ in the risky asset is a decreasing function of the relative risk aversion $\eta$. 
	Economically, that means insurers with higher relative risk aversion invest more conservatively than those with lower relative risk aversion, which is obviously consistent with the definition of risk aversion. 
	
	\item (Optimal liability ratio $\kappa^*$ in the middle panel of Figure \ref{fig:rho}) 
	
	$\kappa^*$ is \emph{not} a monotone function of $\rho$ globally. Instead, there exists a negative threshold $\hat \rho$ ($\hat \rho$ is around -0.15 in the numerical example) such that $\kappa^*$ is increasing over $(\hat \rho,1)$ and decreasing over $(-1, \hat \rho)$. 
	In addition, the increasing part is ``steeper" than the decreasing part, i.e., the same increment of $\rho$ over $(\hat \rho,1)$ has a bigger impact on $\kappa^*$ than the one over $(-1, \hat \rho)$. 
	On the other hand, we observe that $\kappa^*$ is a decreasing function of the risk aversion parameter $\eta$. Therefore, an increase in risk aversion always leads to a decrease in the optimal liability ratio $\kappa^*$.
	
	\item (Optimal dividend rate $\xi^*$ in the bottom panel of Figure \ref{fig:rho}) 
	
	An immediate and also important observation is that the shape of $\xi^*$ differs dramatically over different regions of $\eta$. 
	When $\eta = 1$ (corresponding to log utility), the optimal dividend rate $\xi^*$ is a flat line, equal to $\delta$ ($\delta =0.15$ from Table \ref{tab:para}).
	When $0<\eta < 1$, $\xi^*$ is a concave function of $\rho$, increasing first from -1 to a negative threshold (around -0.25 in Figure \ref{fig:rho}) and decreasing afterwards. 
	When $\eta > 1$, $\xi^*$ is a convex function of $\rho$, decreasing first and then increasing. 
	Furthermore, when $\eta$ increases over $(1, \infty)$, the insurer reduces the optimal dividend rate $\xi^*$. 
\end{itemize}

\begin{figure}
	\begin{center}
		\caption{Impact of the Correlation Coefficient $\rho$ and Risk Aversion $\eta$ on the Optimal Strategies}
		\label{fig:rho}
		\vspace{1ex}
		
		\includegraphics[width = 0.7 \textwidth, height= 0.3\textheight]{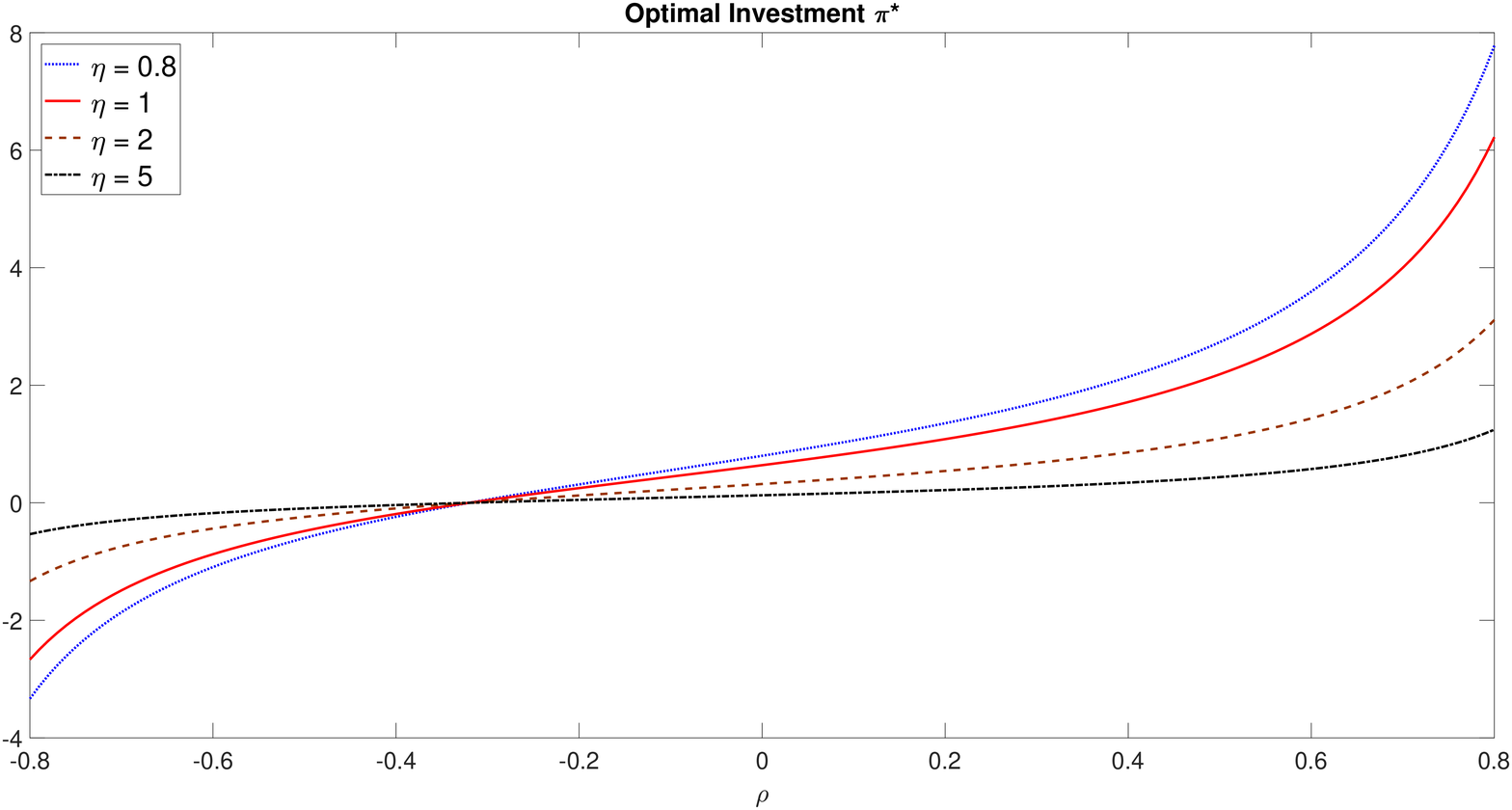}\\
		\includegraphics[width = 0.7 \textwidth, height= 0.3\textheight]{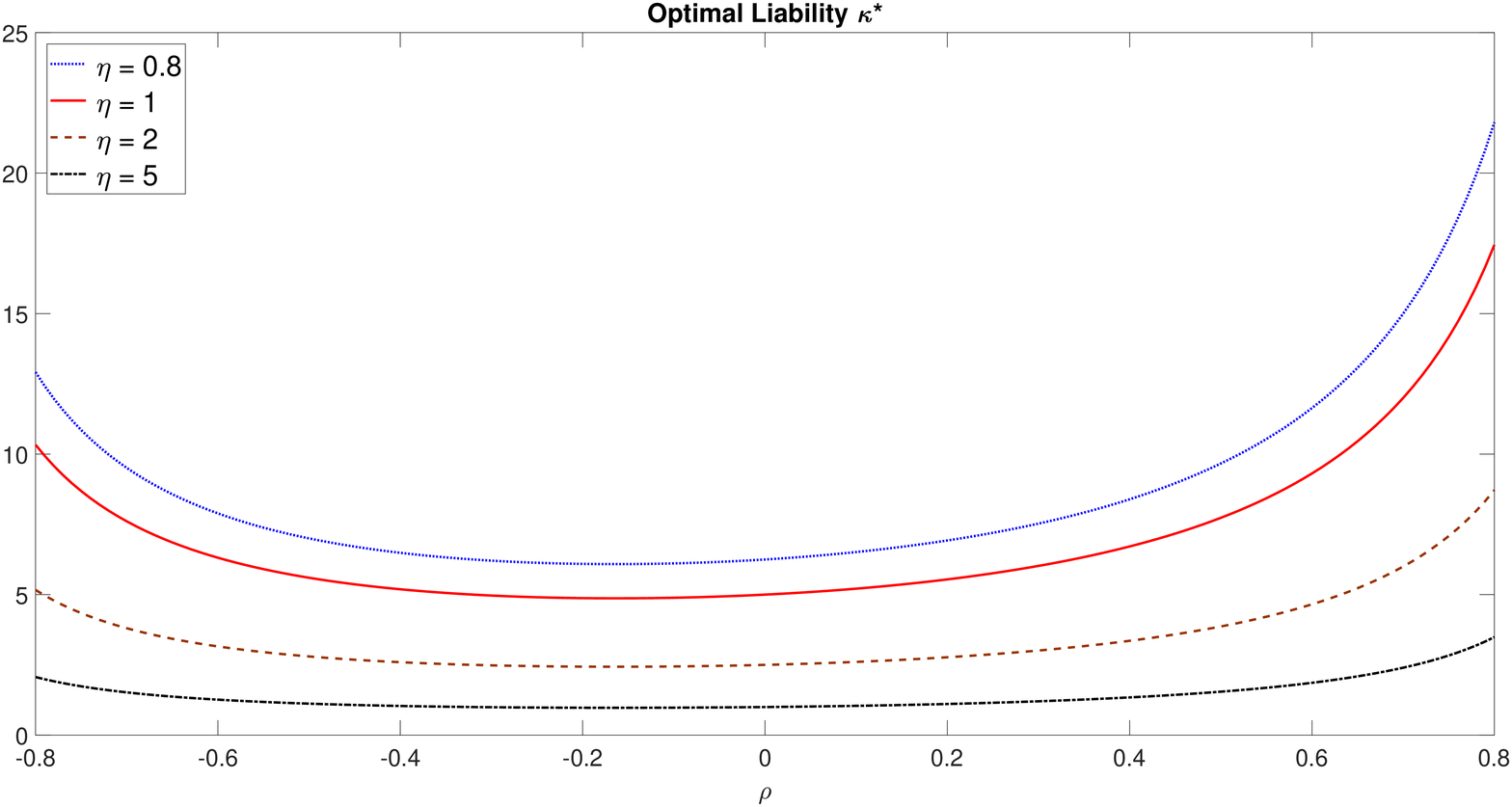}\\
		\includegraphics[width = 0.7 \textwidth, height= 0.3\textheight]{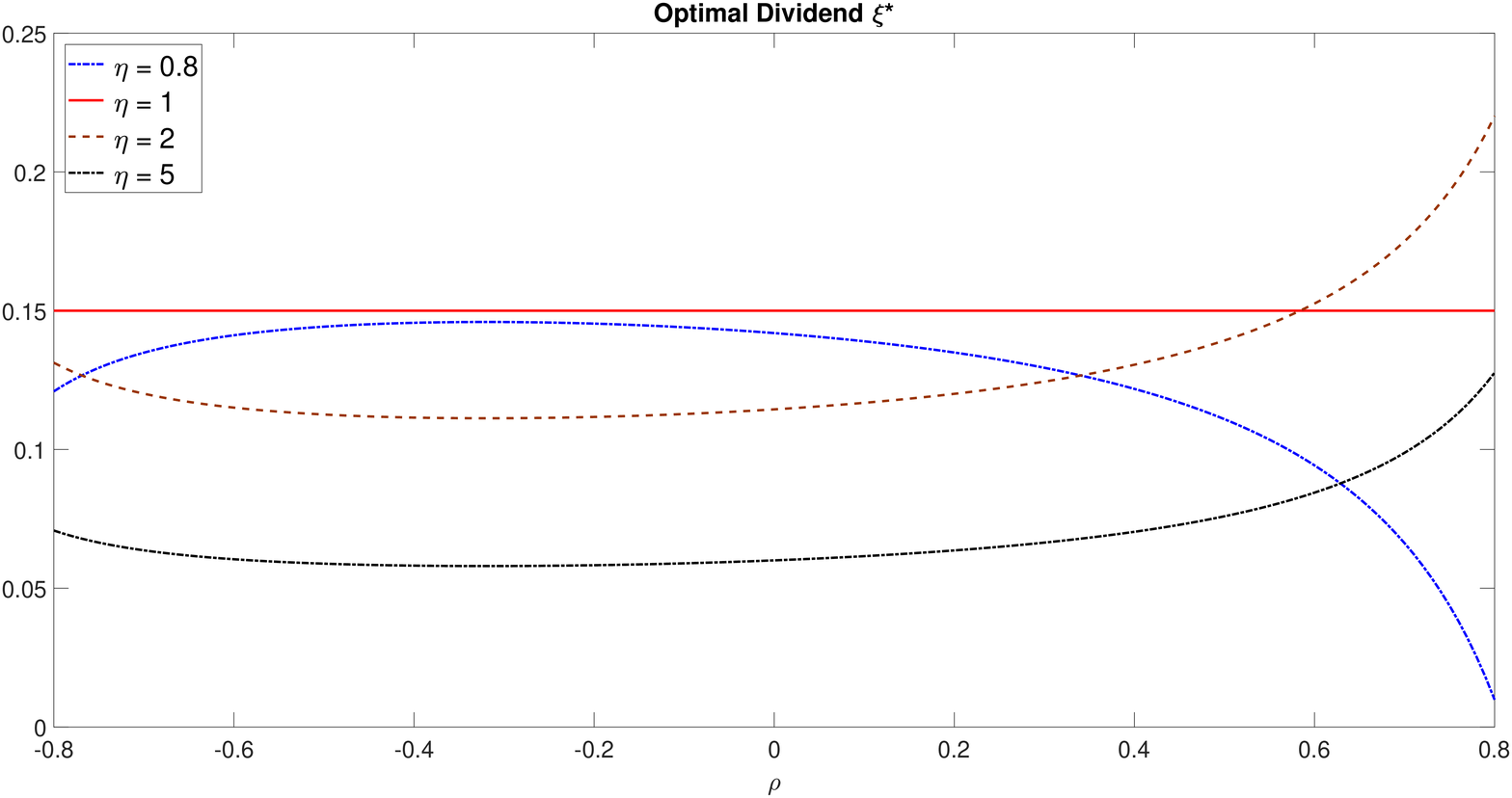}
	\end{center}
	\vspace{-3ex}
	{\footnotesize Note. We plot the insurer's optimal investment $\pi^*$ (upper panel), liability $\kappa^*$ (middle panel), and dividend $\xi^*$ (lower), as a function of the correlation coefficient $\rho$ over $(-0.8, 0.8)$, under four different risk aversion levels $\eta = 0.8$ (dotted blue), $\eta = 1$ (solid red), $\eta = 2$ (dashed brown), and $\eta = 5$ (dash-dot black). The parameters are chosen from Table \ref{tab:para}.}
\end{figure}

Our sensitivity analysis so far is based on the assumption that there are no jumps in the risk process $R$  \eqref{eq:dR}. 
In what follows, we relax this assumption and numerically study how the jump intensity $\lam$ affects the insurer's optimal strategy $u^*$, which is now given by \eqref{eq:pow_cont_op} in Proposition \ref{prop:pow_cont}. 
Due to the inclusion of jumps, we now set the premium rate by the expected value principle: $p = (1+\theta) \times (\alpha + \lam \gam)$, where $\theta = 50\%$.\footnote{Such a rule on premium guarantees that the assumptions in \eqref{eq:psi} hold true. Note that the choice of $p = 0.15$ in Table \ref{tab:para}  is equivalent to $p= (1 + 50\%) \times (\alpha + \lam \gam)$, with $\lam =0$ there.}
We further assume the insurer's risk aversion is $\eta = 2>1$ and fix the jump size $\gam = 0.3$ (3 times the diffusion parameter $\beta = 0.1$).   
We  plot the insurer's optimal investment $\pi^*$, liability ratio $\kappa^*$, and dividend rate $\xi^*$ as a function of $\lam$ over $(0.01,0.2)$ in Figure \ref{fig:lam}. 
The key results are summarized as follows:

	\begin{figure}
	\begin{center}
		\caption{Impact of the Jump Intensity $\lambda$ on the Optimal Strategies}
		\label{fig:lam}
		\vspace{1ex}
		
		\includegraphics[width = 0.7 \textwidth, height= 0.3\textheight]{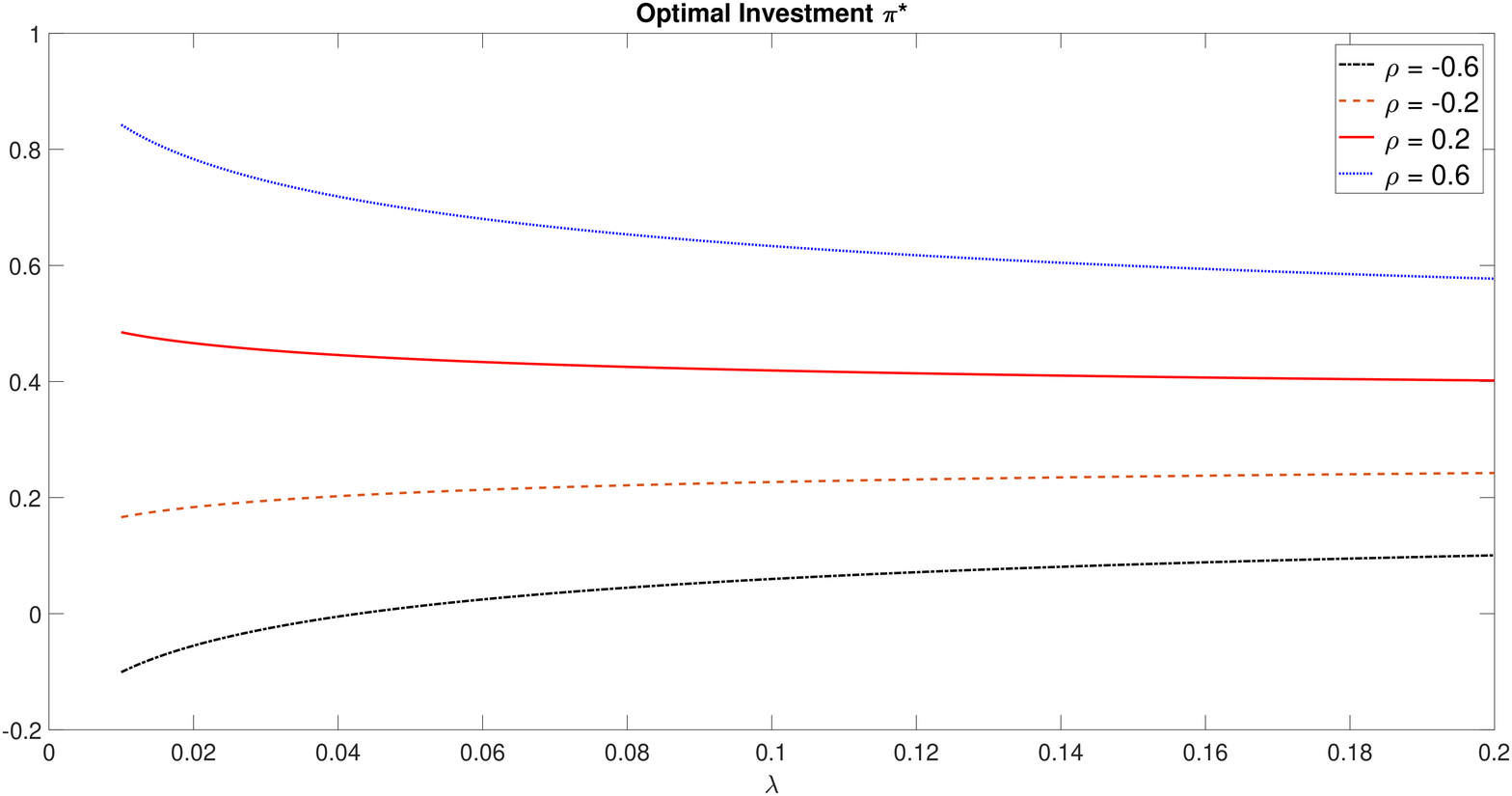}\\
		\includegraphics[width = 0.7 \textwidth, height= 0.3\textheight]{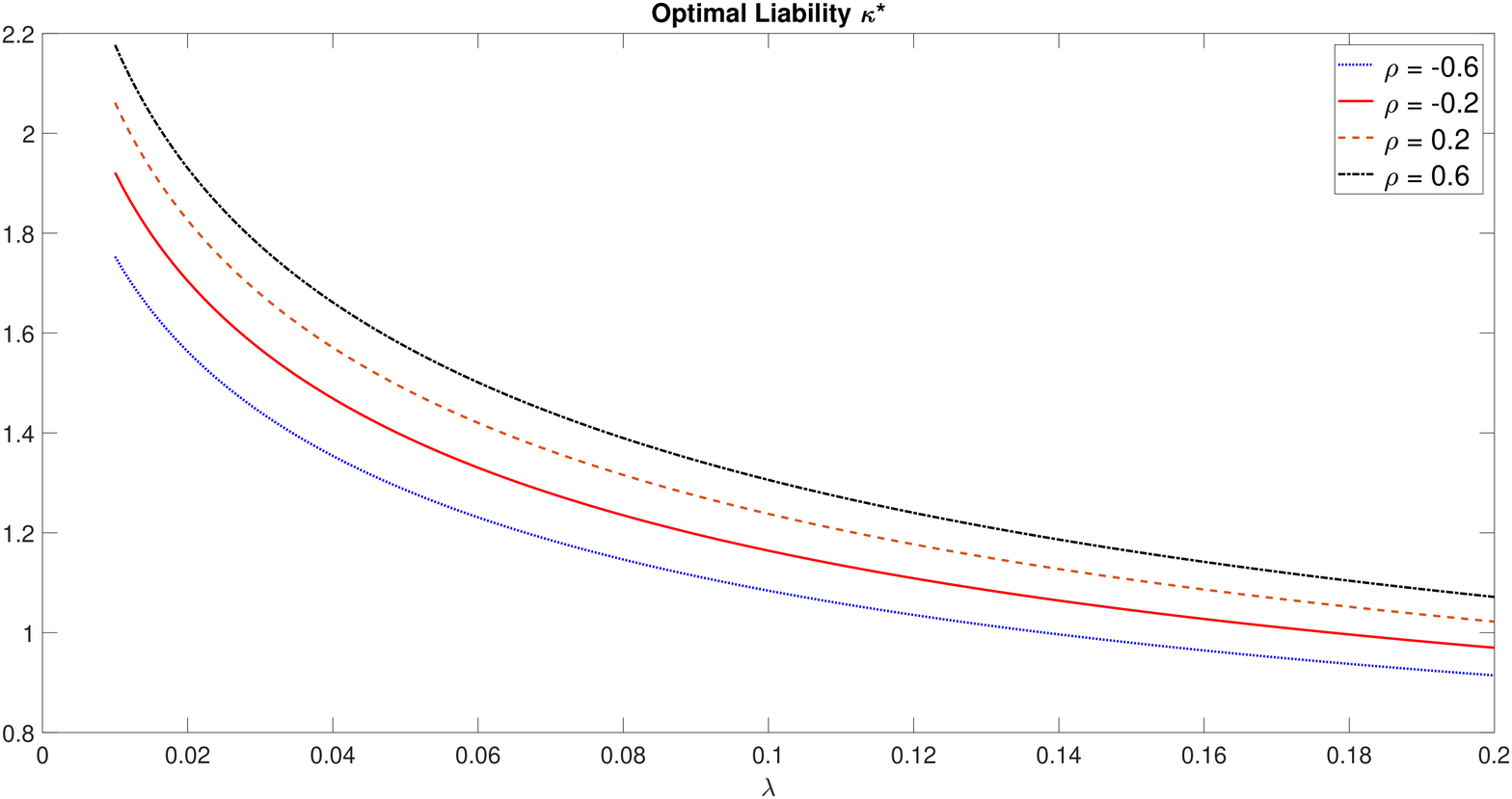}\\
		\includegraphics[width = 0.7 \textwidth, height= 0.3\textheight]{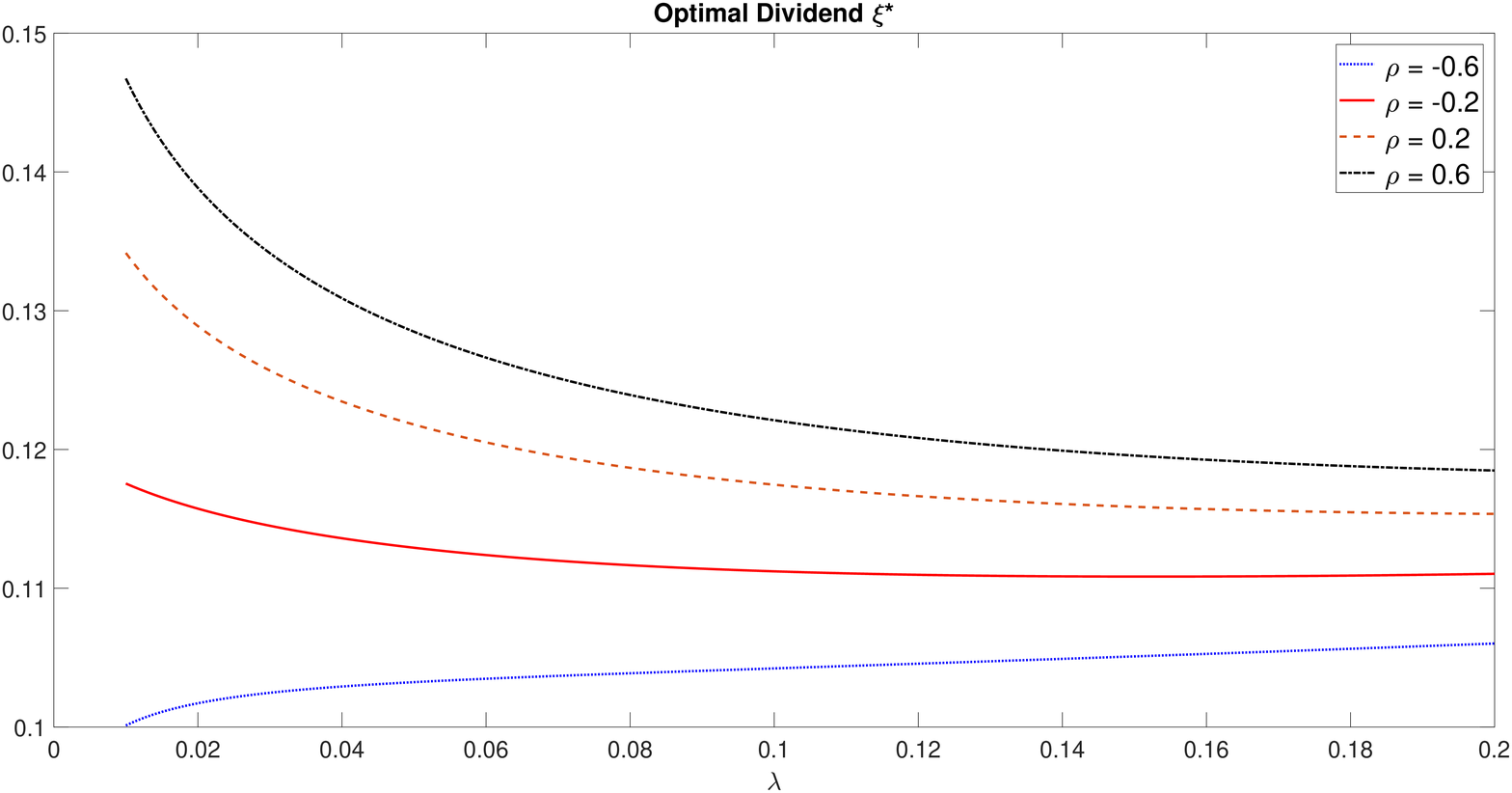}
	\end{center}
	\vspace{-3ex}
	{\footnotesize Note. We plot the insurer's optimal investment $\pi^*$ (upper panel), liability $\kappa^*$ (middle panel), and dividend $\xi^*$ (lower), as a function of the jump intensity $\lam$ over $(0.01, 0.2)$, under four different correlation levels $\rho = -0.6$ (dotted blue), $\rho = -0.2$ (solid red), $\rho = 0.2$ (dashed brown), and $\rho = 0.6$ (dash-dot black). We set $p=1.5(\alpha+\lam \gam)$, $\eta = 2$,  $\gamma = 0.3$, and the rest by Table \ref{tab:para}.}
\end{figure}

\begin{itemize}
	\item (Optimal investment $\pi^*$ in the upper panel of Figure \ref{fig:lam}) 
	
	As anticipated, the results are divided into two regions of the correlation coefficient $\rho$ by a negative threshold $\hat \rho$.
	When $\rho$ is greater than $\hat \rho$, the optimal investment strategy $\pi^*$ is a decreasing function of the jump intensity  $\lam$, and thus the insurer invests less in the risky asset as $\lam$ increases. 
	When $\rho$ is less than $\hat \rho$, $\pi^*$ becomes an increasing function of $\lam$. 
	In particular, for a very negative $\rho$ (e.g., $\rho = -0.6$ in Figure \ref{fig:lam}), $\pi^*$ may even change from a negative value to a positive value when $\lam$ increases. 
	These results can be explained by the expression of $\pi^*$ in \eqref{eq:pow_cont_op}, which has a positive myopic term, independent of $\lam$, and a hedging term with the same sign as $\rho$, which depends on $\lam$ through $\kappa^*$.
	Last, Figure \ref{fig:lam} also confirms the finding from Figure \ref{fig:rho} that $\pi^*$ is an increasing function of $\rho$.
	
	\item (Optimal liability ratio $\kappa^*$ in the middle panel of Figure \ref{fig:lam}) 
	
	We observe that $\kappa^*$ is a decreasing function of $\lam$. 
	The economic explanation for such a result is that, as $\lam$ increases, the insurable risk becomes more risky, and the insurer's rational response is to reduce underwriting. 
	Another observation from the numerical example is that $\kappa^*$ increases as the correlation $\rho$ increases.

	\item (Optimal dividend rate $\xi^*$ in the bottom panel of Figure \ref{fig:lam}) 
	
	%By \eqref{eq:psi} and \eqref{eq:pow_cont_op}, the insurer's optimal dividend strategy $\xi^*$ depends on the optimal investment strategy $\pi^*$ and the optimal liability ratio strategy $\kappa^*$. 
	%As seen above, the impact of $\lam$ on $\kappa^*$ is known (a negative relation), but the impact of $\lam$ on $\pi^*$ is contingent on the value of $\rho$. 
	%Consequently, how $\lam$ affects $\xi^*$ also depends on the value of $\rho$. 
	Similar to the result of $\pi^*$, how the jump intensity $\lam$ affects the optimal dividend strategy $\xi^*$ also depends on the value of $\rho$.
	When $\rho$ is greater than a negative threshold $\hat \rho$, a negative relation between $\xi^*$ and $\lam$ is shown by Figure \ref{fig:lam}. 
	However, for a very negative $\rho$ (e.g., $\rho = -0.6$ in Figure \ref{fig:lam}), it is possible that $\xi^*$ goes up as $\lam$ increases.  
	Following the discussions on $\rho$ in Section \ref{sub:analytic}, let us interpret $\rho > \hat \rho$ as a ``preferable'' market condition and $\rho < \hat \rho$ an ``adverse'' market condition. 
	In a preferable market, although an increase in $\lam$ makes the insurance business more risky, the insurer still wants to maintain the surplus at a relatively high level (since the market is preferable), and to achieve so, the insurer reduces the optimal dividend rate. 
	%To 
	On the contrary, in an adverse market, the insurer pays dividend at a higher rate when $\lam$ increases, that is because keeping a large surplus is no longer attractive when the available business opportunities are unfavorable. 
	In other words, when the market is in adverse condition, the ``utility'' of cashing out dividends now outweighs the ``utility'' of saving for business opportunities in future.
\end{itemize}

\section{Conclusions}
\label{sec:con}

We introduce a combined financial and insurance market consisting of a risk-free asset, a risky asset, 
% with price given by a geometric Brownian motion, 
and a risk process $R$.
The risk process $R$ is modeled by a jump-diffusion process, possibly correlated with the risky asset, and represents the amount of liabilities per policy or per unit in underwriting.
A representative insurer controls simultaneously the investment weight in the risky asset, the liability ratio in underwriting, and the dividend payout rate to shareholders. 
The insurer seeks an optimal triplet strategy to maximize her expected discounted utility of dividend over an infinite horizon. 
We first solve such an optimal dividend problem for \emph{constant} strategies and then apply a perturbation approach to show that the optimal constant strategy remains optimal over all admissible strategies. 
When the insurer's utility is given by a log or a power utility function, we obtain the optimal strategy and the value function in closed form. 
We further utilize the explicit results to conduct economic analyses, both analytically and numerically, to study the impact of various market parameters and risk aversion on the insurer's optimal strategy.     

%\section*{Acknowledgments}
%The research of Z.J. is partially supported by the Research Grants Council of the Hong Kong Special Administrative Region (No. 17330816).
%The research of Z.Q.X.  is partially supported by the National Natural Science Foundation of China  (No. 11971409) and Hong Kong GRF (No. 15204216 and No. 15202817). 
%B.Z. acknowledges a start-up grant  from the University of Connecticut.

\section*{Disclosure statement}
No potential conflict of interest was reported by the author(s).

\bibliographystyle{apalike}
\bibliography{reference}

\begin{thebibliography}{}

\bibitem[Albrecher et~al., 2005]{AlbrecherHT05}
Albrecher, H., Hartinger, J., and Tichy, R.~F. (2005).
\newblock On the distribution of dividend payments and the discounted penalty
  function in a risk model with linear dividend barrier.
\newblock {\em Scandinavian Actuarial Journal}, 2005(2):103--126.

\bibitem[Albrecher and Ivanovs, 2018]{AlbrecherI18}
Albrecher, H. and Ivanovs, J. (2018).
\newblock On the distribution of dividend payments and the discounted penalty
  function in a risk model with linear dividend barrier.
\newblock {\em Scandinavian Actuarial Journal}, 2018(1):76--83.

\bibitem[Albrecher and Thonhauser, 2009]{albrecher2009optimality}
Albrecher, H. and Thonhauser, S. (2009).
\newblock Optimality results for dividend problems in insurance.
\newblock {\em RACSAM-Revista de la Real Academia de Ciencias Exactas,
  F\'isicas y Naturales. Serie A. Matem\'aticas}, 103(2):295--320.

\bibitem[Asmussen et~al., 2000]{asmussen2000optimal}
Asmussen, S., H{\o}jgaard, B., and Taksar, M. (2000).
\newblock Optimal risk control and dividend distribution policies. example of
  excess-of loss reinsurance for an insurance corporation.
\newblock {\em Finance and Stochastics}, 4(3):299--324.

\bibitem[Asmussen and Taksar, 1997]{asmussen1997controlled}
Asmussen, S. and Taksar, M. (1997).
\newblock Controlled diffusion models for optimal dividend pay-out.
\newblock {\em Insurance: Mathematics and Economics}, 20(1):1--15.

\bibitem[Avanzi, 2009]{avanzi2009strategies}
Avanzi, B. (2009).
\newblock Strategies for dividend distribution: A review.
\newblock {\em North American Actuarial Journal}, 13(2):217--251.

\bibitem[Avanzi and Gerber, 2008]{avanzi2008optimal}
Avanzi, B. and Gerber, H.~U. (2008).
\newblock Optimal dividends in the dual model with diffusion.
\newblock {\em ASTIN Bulletin}, 38(2):653--667.

\bibitem[Avanzi et~al., 2007]{avanzi2007optimal}
Avanzi, B., Gerber, H.~U., and Shiu, E.~S. (2007).
\newblock Optimal dividends in the dual model.
\newblock {\em Insurance: Mathematics and Economics}, 41(1):111--123.

\bibitem[Avanzi and Wong, 2012]{avanzi2012mean}
Avanzi, B. and Wong, B. (2012).
\newblock On a mean reverting dividend strategy with brownian motion.
\newblock {\em Insurance: Mathematics and Economics}, 51(2):229--238.

\bibitem[Azcue and Muler, 2005]{azcue2005optimal}
Azcue, P. and Muler, N. (2005).
\newblock Optimal reinsurance and dividend distribution policies in the
  {Cram{\'e}r-Lundberg} model.
\newblock {\em Mathematical Finance}, 15(2):261--308.

\bibitem[Bai and Guo, 2010]{bai2010optimal}
Bai, L. and Guo, J. (2010).
\newblock Optimal dividend payments in the classical risk model when payments
  are subject to both transaction costs and taxes.
\newblock {\em Scandinavian Actuarial Journal}, 2010(1):36--55.

\bibitem[Bayraktar et~al., 2013]{bayraktar2013optimal}
Bayraktar, E., Kyprianou, A.~E., and Yamazaki, K. (2013).
\newblock On optimal dividends in the dual model.
\newblock {\em ASTIN Bulletin}, 43(3):359--372.

\bibitem[Bayraktar et~al., 2014]{bayraktar2014optimal}
Bayraktar, E., Kyprianou, A.~E., and Yamazaki, K. (2014).
\newblock Optimal dividends in the dual model under transaction costs.
\newblock {\em Insurance: Mathematics and Economics}, 54:133--143.

\bibitem[Bernard et~al., 2020]{bernard2020optimal}
Bernard, C., Liu, F., and Vanduffel, S. (2020).
\newblock Optimal insurance in the presence of multiple policyholders.
\newblock {\em Journal of Economic Behavior \& Organization}, forthcoming.

\bibitem[Browne, 1995]{browne1995optimal}
Browne, S. (1995).
\newblock Optimal investment policies for a firm with a random risk process:
  {Exponential} utility and minimizing the probability of ruin.
\newblock {\em Mathematics of Operations Research}, 20(4):937--958.

\bibitem[Cadenillas et~al., 2006]{cadenillas2006classical}
Cadenillas, A., Choulli, T., Taksar, M., and Zhang, L. (2006).
\newblock Classical and impulse stochastic control for the optimization of the
  dividend and risk policies of an insurance firm.
\newblock {\em Mathematical Finance}, 16(1):181--202.

\bibitem[Cadenillas et~al., 2007]{cadenillas2007optimal}
Cadenillas, A., Sarkar, S., and Zapatero, F. (2007).
\newblock Optimal dividend policy with mean-reverting cash reservoir.
\newblock {\em Mathematical Finance}, 17(1):81--109.

\bibitem[Capponi and Figueroa-L{\'o}pez, 2014]{capponi2014dynamic}
Capponi, A. and Figueroa-L{\'o}pez, J.~E. (2014).
\newblock Dynamic portfolio optimization with a defaultable security and
  regime-switching.
\newblock {\em Mathematical Finance}, 24(2):207--249.

\bibitem[Chen et~al., 2014]{chen2014optimal}
Chen, S., Li, Z., and Zeng, Y. (2014).
\newblock Optimal dividend strategies with time-inconsistent preferences.
\newblock {\em Journal of Economic Dynamics and Control}, 46:150--172.

\bibitem[Choulli et~al., 2003]{choulli2003diffusion}
Choulli, T., Taksar, M., and Zhou, X.~Y. (2003).
\newblock A diffusion model for optimal dividend distribution for a company
  with constraints on risk control.
\newblock {\em SIAM Journal on Control and Optimization}, 41(6):1946--1979.

\bibitem[De~Finetti, 1957]{de1957impostazione}
De~Finetti, B. (1957).
\newblock Su un’impostazione alternativa della teoria collettiva del rischio.
\newblock In {\em Transactions of the XVth International Congress of
  Actuaries}, volume~2, pages 433--443. New York.

\bibitem[Fleming and Soner, 2006]{fleming2006controlled}
Fleming, W.~H. and Soner, H.~M. (2006).
\newblock {\em Controlled Markov Processes and Viscosity Solutions}.
\newblock Springer Science \& Business Media.

\bibitem[Gerber and Shiu, 2004]{gerber2004optimal}
Gerber, H.~U. and Shiu, E.~S. (2004).
\newblock Optimal dividends: analysis with brownian motion.
\newblock {\em North American Actuarial Journal}, 8(1):1--20.

\bibitem[Grandits et~al., 2007]{grandits2007optimal}
Grandits, P., Hubalek, F., Schachermayer, W., and Žigo, M. (2007).
\newblock Optimal expected exponential utility of dividend payments in a
  brownian risk model.
\newblock {\em Scandinavian Actuarial Journal}, 2007(2):73--107.

\bibitem[Gu et~al., 2018]{gu2018optimal}
Gu, J.-W., Steffensen, M., and Zheng, H. (2018).
\newblock Optimal dividend strategies of two collaborating businesses in the
  diffusion approximation model.
\newblock {\em Mathematics of Operations Research}, 43(2):377--398.

\bibitem[Herdegen et~al., 2020]{herdegen2020elementary}
Herdegen, M., Hobson, D., and Jerome, J. (2020).
\newblock An elementary approach to the merton problem.
\newblock {\em arXiv preprint arXiv:2006.05260}.

\bibitem[H{\o}jgaard and Taksar, 1998]{hojgaard1998optimal}
H{\o}jgaard, B. and Taksar, M. (1998).
\newblock Optimal proportional reinsurance policies for diffusion models.
\newblock {\em Scandinavian Actuarial Journal}, 1998(2):166--180.

\bibitem[H{\o}jgaard and Taksar, 1999]{hojgaard1999controlling}
H{\o}jgaard, B. and Taksar, M. (1999).
\newblock Controlling risk exposure and dividends payout schemes: insurance
  company example.
\newblock {\em Mathematical Finance}, 9(2):153--182.

\bibitem[Jeanblanc-Picqu{\'e} and Shiryaev, 1995]{jeanblanc1995optimization}
Jeanblanc-Picqu{\'e}, M. and Shiryaev, A.~N. (1995).
\newblock Optimization of the flow of dividends.
\newblock {\em Uspekhi Matematicheskikh Nauk}, 50(2):25--46.

\bibitem[Jin et~al., 2020]{jin2020optimal}
Jin, Z., Liao, H., Yang, Y., and Yu, X. (2020).
\newblock Optimal dividend strategy for an insurance group with contagious
  default risk.
\newblock {\em Scandinavian Actuarial Journal}, forthcoming.

\bibitem[Jin et~al., 2013]{jin2013numerical}
Jin, Z., Yang, H., and Yin, G. (2013).
\newblock Numerical methods for optimal dividend payment and investment
  strategies of regime-switching jump diffusion models with capital injections.
\newblock {\em Automatica}, 49(8):2317--2329.

\bibitem[Jin et~al., 2015]{jin2015optimal}
Jin, Z., Yang, H., and Yin, G. (2015).
\newblock Optimal debt ratio and dividend payment strategies with reinsurance.
\newblock {\em Insurance: Mathematics and Economics}, 64:351--363.

\bibitem[Karatzas et~al., 1986]{karatzas1986explicit}
Karatzas, I., Lehoczky, J.~P., Sethi, S.~P., and Shreve, S.~E. (1986).
\newblock Explicit solution of a general consumption/investment problem.
\newblock {\em Mathematics of Operations Research}, 11(2):261--294.

\bibitem[Karatzas and Shreve, 1998]{karatzas1998methods}
Karatzas, I. and Shreve, S.~E. (1998).
\newblock {\em Methods of Mathematical Finance}.
\newblock Springer.

\bibitem[Kulenko and Schmidli, 2008]{kulenko2008optimal}
Kulenko, N. and Schmidli, H. (2008).
\newblock Optimal dividend strategies in a cram{\'e}r--lundberg model with
  capital injections.
\newblock {\em Insurance: Mathematics and Economics}, 43(2):270--278.

\bibitem[Lapham et~al., 1996]{lapham1996genetic}
Lapham, E.~V., Kozma, C., and Weiss, J.~O. (1996).
\newblock Genetic discrimination: Perspectives of consumers.
\newblock {\em Science}, 274(5287):621--624.

\bibitem[Lindensj{\"o} and Lindskog, 2020]{lindensjo2020optimal}
Lindensj{\"o}, K. and Lindskog, F. (2020).
\newblock Optimal dividends and capital injection under dividend restrictions.
\newblock {\em Mathematical Methods of Operations Research}, pages 1--27.

\bibitem[Merton, 1969]{merton1969lifetime}
Merton, R.~C. (1969).
\newblock Lifetime portfolio selection under uncertainty: The continuous-time
  case.
\newblock {\em Review of Economics and Statistics}, 51(3):247--257.

\bibitem[Merton, 1971]{merton1971optimum}
Merton, R.~C. (1971).
\newblock Optimum consumption and portfolio rules in a continuous-time model.
\newblock {\em Journal of Economic Theory}, 3(4):373--413.

\bibitem[Meyer and Meyer, 2005]{meyer2005relative}
Meyer, D.~J. and Meyer, J. (2005).
\newblock Relative risk aversion: What do we know?
\newblock {\em Journal of Risk and Uncertainty}, 31(3):243--262.

\bibitem[{\O}ksendal and Sulem, 2005]{oksendal2005applied}
{\O}ksendal, B.~K. and Sulem, A. (2005).
\newblock {\em Applied Stochastic Control of Jump Diffusions}.
\newblock Springer.

\bibitem[Schmidli, 2007]{schmidli2007stochastic}
Schmidli, H. (2007).
\newblock {\em Stochastic Control in Insurance}.
\newblock Springer Science \& Business Media.

\bibitem[Schmidli, 2017]{schmidli2017capital}
Schmidli, H. (2017).
\newblock On capital injections and dividends with tax in a diffusion
  approximation.
\newblock {\em Scandinavian Actuarial Journal}, 2017(9):751--760.

\bibitem[Shen and Zou, 2021]{shen2020mean}
Shen, Y. and Zou, B. (2021).
\newblock Mean-varince investment and risk control strategies: {A}
  time-consistant approach via a forward anxiliary process.
\newblock {\em Insurance: Mathematics and Economics}, 97:68--80.

\bibitem[Sotomayor and Cadenillas, 2011]{sotomayor2011classical}
Sotomayor, L.~R. and Cadenillas, A. (2011).
\newblock Classical and singular stochastic control for the optimal dividend
  policy when there is regime switching.
\newblock {\em Insurance: Mathematics and Economics}, 48(3):344--354.

\bibitem[Stein, 2012]{stein2012stochastic}
Stein, J.~L. (2012).
\newblock {\em Stochastic Optimal Control and the US Financial Debt Crisis}.
\newblock Springer.

\bibitem[Taksar and Zhou, 1998]{taksar1998optimal}
Taksar, M.~I. and Zhou, X.~Y. (1998).
\newblock Optimal risk and dividend control for a company with a debt
  liability.
\newblock {\em Insurance: Mathematics and Economics}, 22(1):105--122.

\bibitem[Thonhauser and Albrecher, 2007]{thonhauser2007dividend}
Thonhauser, S. and Albrecher, H. (2007).
\newblock Dividend maximization under consideration of the time value of ruin.
\newblock {\em Insurance: Mathematics and Economics}, 41(1):163--184.

\bibitem[Thonhauser and Albrecher, 2011]{thonhauser2011optimal}
Thonhauser, S. and Albrecher, H. (2011).
\newblock Optimal dividend strategies for a compound poisson process under
  transaction costs and power utility.
\newblock {\em Stochastic Models}, 27(1):120--140.

\bibitem[Xu et~al., 2020]{xu2020optimal}
Xu, L., Yao, D., and Cheng, G. (2020).
\newblock Optimal investment and dividend for an insurer under a markov regime
  switching market with high gain tax.
\newblock {\em Journal of Industrial \& Management Optimization}, 16(1):325.

\bibitem[Yao et~al., 2011]{yao2011optimal}
Yao, D., Yang, H., and Wang, R. (2011).
\newblock Optimal dividend and capital injection problem in the dual model with
  proportional and fixed transaction costs.
\newblock {\em European Journal of Operational Research}, 211(3):568--576.

\bibitem[Zhou and Jin, 2020]{zhou2020optimal}
Zhou, Z. and Jin, Z. (2020).
\newblock Optimal equilibrium barrier strategies for time-inconsistent dividend
  problems in discrete time.
\newblock {\em Insurance: Mathematics and Economics}, 94:100--108.

\bibitem[Zhu, 2013]{zhu2013optimal}
Zhu, J. (2013).
\newblock Optimal dividend control for a generalized risk model with investment
  incomes and debit interest.
\newblock {\em Scandinavian Actuarial Journal}, 2013(2):140--162.

\bibitem[Zhu, 2015]{zhu2015dividend}
Zhu, J. (2015).
\newblock Dividend optimization for general diffusions with restricted dividend
  payment rates.
\newblock {\em Scandinavian Actuarial Journal}, 2015(7):592--615.

\bibitem[Zhu et~al., 2020]{zhu2020singular}
Zhu, J., Siu, T.~K., and Yang, H. (2020).
\newblock Singular dividend optimization for a linear diffusion model with
  time-inconsistent preferences.
\newblock {\em European Journal of Operational Research}, 285(1):66--80.

\bibitem[Zou and Cadenillas, 2014]{zou2014optimal}
Zou, B. and Cadenillas, A. (2014).
\newblock Optimal investment and risk control policies for an insurer: Expected
  utility maximization.
\newblock {\em Insurance: Mathematics and Economics}, 58:57--67.

\end{thebibliography}

\end{document}